\documentclass[
  aps,
  pra,
  reprint,            
  superscriptaddress, 
  amsmath,amssymb,
  notitlepage
]{revtex4-2}
\usepackage[utf8]{inputenc}
\usepackage[T1]{fontenc}
\usepackage{lmodern}

\usepackage{hyperref}       
\usepackage{url}            
\usepackage{booktabs}       
\usepackage{amsfonts}
\usepackage{dsfont}       
\usepackage{nicefrac}       
\usepackage{microtype}      

\usepackage{graphicx}       
\graphicspath{{media/}}     

\usepackage[normalem]{ulem}
\usepackage{array}
\usepackage{tikz}
\usetikzlibrary{positioning,quotes}
\usepackage{subcaption}     
\usepackage{xcolor}
\usepackage{soul}
\usepackage{amsthm}
\usepackage{framed}
\usepackage{mathtools}
\usepackage{quantikz}
\usepackage{cancel}
\usepackage{algpseudocode}

\theoremstyle{plain}
\newtheorem{theorem}{Theorem}        
\newtheorem{proposition}{Proposition}
\newtheorem{corollary}{Corollary}    

\theoremstyle{definition}

\theoremstyle{definition}
\newtheorem{definition}{Definition}  
\newtheorem{algorithm}{Algorithm}

\usepackage{hyperref}
\usepackage{setspace} 
\usepackage{graphicx}

\begin{document}
%

\title{Analytical controls for dispersive multi-qubit interactions \\ in mediator-coupled quantum registers}

\author{Bora Baran}
\email{b.baran@fz-juelich.de}
\affiliation{Forschungszentrum Jülich, Institute of Quantum Control,%
 Peter Grünberg Institut (PGI-8), 52425 Jülich, Germany}
\affiliation{Institute for Theoretical Physics, University of Cologne, Zuelpicher Strasse 77, 50937 Cologne, Germany}

\author{Tommaso Calarco}
\affiliation{Forschungszentrum Jülich, Institute of Quantum Control,%
 Peter Grünberg Institut (PGI-8), 52425 Jülich, Germany}
\affiliation{Institute for Theoretical Physics, University of Cologne, Zuelpicher Strasse 77, 50937 Cologne, Germany}
\affiliation{Dipartimento di Fisica e Astronomia, Università di Bologna, 40127 Bologna, Italy}

\author{James O'Sullivan}
\affiliation{Quantronics Group, Service de Physique de l’Etat Condensé (CNRS, UMR 3680), IRAMIS, CEA-Saclay, Université Paris-Saclay, 91191 Gif-sur-Yvette, France}

\author{Matthias M. Müller}
\affiliation{Forschungszentrum Jülich, Institute of Quantum Control,%
Peter Grünberg Institut (PGI-8), 52425 Jülich, Germany}

\author{Felix Motzoi}
\affiliation{Forschungszentrum Jülich, Institute of Quantum Control,%
 Peter Grünberg Institut (PGI-8), 52425 Jülich, Germany}
\affiliation{Institute for Theoretical Physics, University of Cologne, Zuelpicher Strasse 77, 50937 Cologne, Germany}

\begin{abstract}
\noindent
Fast control of long-lived quantum registers often relies on an auxiliary
degree of freedom, a \emph{mediator}, but real mediator excitation can
introduce decoherence, leakage, or loss when the mediator is not the
coherence-optimized component. We introduce \emph{dispersive interaction via
analytical linear inversion} (DIAL), an analytical control construction for
excitation-suppressed interaction engineering in mediator-coupled quantum
registers. In the dispersive regime, we show that off-resonant driving of an effective two-level mediator
transition yields a diagonal Hamiltonian whose Pauli-\(Z\) interaction
rates depend linearly on the drive intensities through a transfer matrix
determined by the register-resolved transition spectrum.
This establishes a direct linear map from experimentally accessible
controls to induced multi-qubit interactions, enabling control design
without iterative propagation of the full driven dynamics.
We turn this structure into the \emph{target-aware DIAL algorithm}, which selects
favorable off-resonant tones and determines their drive
intensities for a prescribed target interaction. 
We benchmark the resulting controls across 100 register realizations for each register size \(n=2,3,4,5\), targeting the highest-order interaction \(Z_1\cdots Z_n\) with an interaction phase of magnitude \(\pi/4\), which generates an \(n\)-qubit GHZ state up to local equivalence.
Propagation under the full driven mediator-register dynamics yields typical gate fidelities above \(0.997\), while the
maximum transient mediator excitation remains near \(1\%\) on average, far below what is typically possible with resonant Rabi drives. 
The manuscript is accompanied by an open-source Python package implementing the target-aware DIAL algorithm and the associated numerical validation workflow.
\end{abstract}

\maketitle

\begingroup
\let\oldnumberline\numberline
\renewcommand{\numberline}[1]{\oldnumberline{#1}\hspace{0.35em}}
\tableofcontents
\endgroup
\section{Introduction}

Long-lived quantum registers are attractive for storing quantum
information, but their isolation from the environment can make direct
multi-qubit control, such as entangling operations, slow or technically
demanding. Hybrid architectures address this limitation by coupling the
protected memory to a more readily addressable auxiliary degree of
freedom, to which we refer as the \emph{mediator}. Depending on the
platform, the mediator may be an electronic degree of freedom, a bosonic
or collective mode, or an ancillary particle,
and may additionally provide initialization, readout, and conditional
control
\cite{kane1998silicon,baart2017coherent,
osullivan2025individual,atzori2018twoqubit,cirac1995quantum,
blais2021circuit,aspelmeyer2014cavity,tabuchi2015coherent,
negnevitsky2018repeated}.
The mediator then provides controllable access to otherwise weakly
coupled or indirectly addressable register degrees of freedom.

This register-mediator strategy has been realized across a wide range
of physical platforms. The nitrogen-vacancy (NV) center in diamond is
a paradigmatic example: there, the electron spin initializes, controls,
and reads out nearby nuclear spins
\cite{neumann2008multipartite,
neumann2010quantum,fuchs2011quantum,vandersar2012decoherence,
gaebel2006room,taminiau2012detection,maurer2012roomtemp,
taminiau2014universal,waldherr2014quantum,cramer2016repeated,
bradley2019ten,abobeih2022fault}, and the same
physics underlies nanoscale magnetometry
\cite{balasubramanian2008nanoscale,maze2008nanoscale,
schirhagl2014nitrogen,casola2018probing}. Dynamical-decoupling
techniques can further protect the register spins while the mediator
is being driven
\cite{viola1999dynamical,delange2010universal,bargill2012suppression}.

Rare-earth ions and molecular spins doped into crystalline hosts follow
a similar logic, pairing an optical transition with long-lived nuclear
or electronic spin storage. Erbium, europium, and other lanthanide
dopants have been employed in this capacity
\cite{yin2013optical,rippe2008experimental,kolesov2012optical,
zhong2015optically,siyushev2014coherent,zhong2018nanophotonic,
zhong2019emerging,serrano2018all,kindem2020control,
chen2020parallel,ruskuc2022nuclear,kuppusamy2024spin}, most recently
in a demonstration of individual erbium-coupled nuclear-spin qubits
with coherence exceeding a second and all-microwave gates driven by
stimulated Raman transitions
\cite{osullivan2025individual}. Optically addressable molecular spins
provide a further realization of this architecture
\cite{atzori2018twoqubit}, and selective optical
addressing of individual nuclear spins in these registers can itself
be mediated by superhyperfine coupling to the host lattice
\cite{car2018selective}.

Semiconductor spin qubits constitute a further electron-mediated
register. In the original Kane proposal, gate-controlled electron
wavefunctions mediate exchange coupling between donor nuclear-spin
qubits in silicon
\cite{kane1998silicon,hollenberg2004charge}, an approach later
realized with a single addressable donor electron
\cite{pla2012single}, while gate-defined quantum dots instead use a
tunable, electron-mediated exchange interaction to entangle spin
qubits directly
\cite{baart2017coherent,malinowski2019fast}. In circuit quantum
electrodynamics, transmon-type superconducting qubits
\cite{koch2007charge}, developed from earlier charge-qubit designs
\cite{vion2002manipulating}, play the register role, with a cavity or
resonator photon serving as the mediator.

The register-mediator concept is not restricted to electronic or spin
degrees of freedom. Trapped-ion registers were originally proposed
around a shared motional mode acting as a quantum bus
\cite{cirac1995quantum}; driving this mode off resonance, in the
low-phonon-occupation limit, mediates spin-spin interactions
\cite{sorensen2000entanglement} and can engineer tunable Ising- and
Heisenberg-like couplings
\cite{porras2004effective}, as demonstrated with hundreds of spins
\cite{britton2012engineered} and reviewed more broadly for programmable
trapped-ion simulators
\cite{monroe2021programmable}. An ancillary ion of a different species
can likewise carry information from data qubits to a spectrally
isolated readout ion
\cite{negnevitsky2018repeated}.

Neutral-atom arrays provide a complementary interaction-mediated architecture,
in which entangling operations exploit the strong and controllable interactions
between transiently populated Rydberg states
\cite{saffman2010quantum}.

Other collective bosonic excitations can
play an analogous mediator role: mechanical modes coupled to light
via radiation pressure in cavity optomechanics
\cite{aspelmeyer2014cavity}, and magnons in ferromagnetic crystals
coherently coupled to superconducting qubits through a virtual cavity
photon
\cite{tabuchi2015coherent,lachancequirion2019hybrid}.

Despite their different microscopic realizations, these architectures
provide several routes by which a mediator can act conditionally on a
quantum register. The analytical construction developed here applies
specifically when an addressable auxiliary transition has a frequency
that depends on the register configuration and approximately preserves
that configuration. In this regime, the register-dependent transition
spectrum can generate conditional energy shifts without simultaneously
mixing the register states. Which physical degree of freedom provides
such a transition, and whether this regime is available, depends on the
platform.

The physical origin of this selective regime is platform dependent. In
electron-nuclear spin registers, manifold-dependent transverse hyperfine
terms can tilt the register quantization axes, while a standard
Zeeman-dominated high-field or secular regime suppresses this mismatch
when the relevant Zeeman splittings exceed the transverse couplings
\cite{taminiau2014universal,bradley2019ten,tosi2017silicon}. In
rare-earth and molecular spin systems, the corresponding selectivity can
instead follow from hyperfine-resolved transitions, selection rules, and
an appropriate transition, orientation, or operating point
\cite{osullivan2025individual,atzori2018twoqubit}. Exchange- or
Ising-coupled ancilla spins can acquire conditional resonance shifts from
neighboring register spins
\cite{yoneda2020qnd,wunderlich2002conditional,nagies2025higherorder}.
For cavity and magnon modes, dispersive coupling can instead produce
state-dependent shifts of the auxiliary-mode resonance
\cite{blais2021circuit,lachancequirion2017resolving}.

Table~\ref{tab:mediator_platforms} summarizes representative
register-mediator realizations and physical origins of the
register-dependent response used in the analytical construction below.
The entries identify candidate routes rather than complete
platform-specific implementations.  

\begin{table*}[t!]
\centering
\small
\caption{
Representative register-mediator realizations and physical origins of
register dependence relevant to the analytical construction below. The
entries identify candidate routes rather than complete platform-specific
implementations.
}
\label{tab:mediator_platforms}
\begin{tabular}{@{}llll@{}}
\toprule

\parbox[t]{0.20\textwidth}{Representative setting}
&
\parbox[t]{0.19\textwidth}{Mediator}
&
\parbox[t]{0.29\textwidth}{Origin of register dependence}
&
\parbox[t]{0.20\textwidth}{Representative references}
\\
\midrule

\parbox[t]{0.20\textwidth}{NV-center nuclear-spin registers}
&
\parbox[t]{0.19\textwidth}{NV electron spin}
&
\parbox[t]{0.29\textwidth}{Hyperfine-resolved nuclear-spin configuration}
&
\parbox[t]{0.20\textwidth}{
\cite{taminiau2014universal,bradley2019ten}}
\\[0.6em]

\parbox[t]{0.20\textwidth}{Donor electron--nuclear-spin registers}
&
\parbox[t]{0.19\textwidth}{Donor electron}
&
\parbox[t]{0.29\textwidth}{Hyperfine-resolved nuclear-spin configuration}
&
\parbox[t]{0.20\textwidth}{
\cite{kane1998silicon,pla2012single,tosi2017silicon}}
\\[0.6em]

\parbox[t]{0.20\textwidth}{Gate-defined quantum-dot spin registers}
&
\parbox[t]{0.19\textwidth}{Ancilla electron spin}
&
\parbox[t]{0.29\textwidth}{Exchange- / Ising-induced conditional spin resonance}
&
\parbox[t]{0.20\textwidth}{
\cite{yoneda2020qnd,baart2017coherent,malinowski2019fast}}
\\[0.6em]

\parbox[t]{0.20\textwidth}{Rare-earth-coupled nuclear-spin registers}
&
\parbox[t]{0.19\textwidth}{Rare-earth electron spin}
&
\parbox[t]{0.29\textwidth}{Hyperfine-resolved nuclear configuration; transition- and orientation-selective}
&
\parbox[t]{0.20\textwidth}{
\cite{osullivan2025individual,ruskuc2022nuclear,kindem2020control}}
\\[0.6em]

\parbox[t]{0.20\textwidth}{Molecular nuclear-spin registers}
&
\parbox[t]{0.19\textwidth}{Molecular electron spin}
&
\parbox[t]{0.29\textwidth}{Hyperfine-resolved nuclear configuration and selection rules}
&
\parbox[t]{0.20\textwidth}{
\cite{atzori2018twoqubit}}
\\[0.6em]

\parbox[t]{0.20\textwidth}{Circuit QED}
&
\parbox[t]{0.19\textwidth}{Cavity / resonator mode}
&
\parbox[t]{0.29\textwidth}{Qubit-state-dependent dispersive resonance shift}
&
\parbox[t]{0.20\textwidth}{
\cite{schuster2005ac,goerz2017charting,deng2025long}}
\\[0.6em]

\parbox[t]{0.20\textwidth}{Trapped-ion spin registers}
&
\parbox[t]{0.19\textwidth}{Auxiliary ion}
&
\parbox[t]{0.29\textwidth}{Magnetic-gradient-induced Ising coupling / conditional spin resonance}
&
\parbox[t]{0.20\textwidth}{
\cite{wunderlich2002conditional,nagies2025higherorder}}
\\[0.6em]

\parbox[t]{0.20\textwidth}{Hybrid magnonic systems}
&
\parbox[t]{0.19\textwidth}{Magnon mode}
&
\parbox[t]{0.29\textwidth}{Qubit-state-dependent dispersive magnon shift}
&
\parbox[t]{0.20\textwidth}{
\cite{tabuchi2015coherent,lachancequirion2017resolving,lachancequirion2019hybrid}}
\\

\bottomrule
\end{tabular}
\end{table*}

\noindent
In many such architectures the mediator and the register play distinct
roles: the register can provide comparatively long-lived storage, while
the mediator provides controllability or an interface to external fields
\cite{waldherr2014quantum,taminiau2014universal,serrano2018all}.
As a result, real mediator excitation can be undesirable: depending on
the platform, it may expose the operation to mediator dephasing,
relaxation, optical decay, ionization
\cite{aslam2013photo}, or leakage outside the intended computational subspace~\cite{blais2021circuit, motzoi2009simple, theis2018counteracting, li2024experimental}.
This motivates control schemes that retain the conditional action of
the mediator while suppressing its population.

A natural route is therefore to operate dispersively. Driving an
auxiliary transition or mode off resonance suppresses real excitation
while still producing state-dependent energy-level shifts, which
effective-Hamiltonian methods such as adiabatic elimination or
Schrieffer-Wolff transformations capture directly
\cite{brion2007adiabatic,schrieffer1966relation,gambetta2011analytic,
bravyi2011schrieffer,reiter2012effective,sanz2015beyond,blais2021circuit},
as well as non-perturbative generalizations
\cite{li2022nonperturbative,jerger2024dispersive}.
Off-resonant energy shifts of this kind are a basic mechanism in atomic
physics and quantum optics
\cite{cohen1992laser,cohen2024atom,grimm2000optical}, and play a
central role in circuit quantum electrodynamics
\cite{schuster2005ac,
goerz2017charting, deng2025long,zhu2013circuit,blais2021circuit}, where they
underlie both dispersive readout and selective, frequency-resolved
control of nominally degenerate transitions
\cite{motzoi2013,degroot2010selective,chow2011simple,
sheldon2016procedure}, building on a longer tradition of selective,
frequency-resolved pulse design in nuclear magnetic resonance
\cite{vandersypen2004nmr}. The control problem is therefore to exploit
these dispersive shifts for selective control operations while avoiding
the error channels introduced by real mediator population.

\noindent
For mediator-coupled registers, standard numerical quantum optimal
control can already search for control pulses that implement target
gates on the register through the full register-mediator dynamics
\cite{rabitz2000whither,khaneja2001time,werschnik2007quantum,khaneja2005optimal, baran2026interaction,glaser2015training,doria2011optimal,caneva2011chopped,reich2012monotonically,koch2022quantum,mahesh2022quantum}. However, it typically treats synthesis as an
iterative, largely black-box forward search over a high-dimensional
control landscape. In this formulation, the dependence of the final
operation on the applied control pulse is generally nonlinear, and the
resulting pulse does not directly reveal how the available spectral
responses generate the desired interaction or unwanted spectator terms.

Here we take a complementary route by exploiting the
configuration-preserving dispersive structure directly. For the
analytical construction, we focus on an addressable transition between
two relevant mediator states that can be treated as an effective
two-level subsystem. Register-state-dependent transition frequencies
then produce configuration-dependent off-resonant energy shifts, while
configuration preservation makes these shifts diagonal in the register
basis. They can therefore be decomposed into distinct Pauli-\(Z\) interaction channels~\cite{baran2026interaction}. For fixed tone frequencies and constant-amplitude multi-tone driving,
the corresponding interaction rates depend linearly on the drive powers.
Depending on the platform, such multi-tone control may be implemented
through microwave or radio-frequency fields, Raman drives, or coherent
driving of bosonic modes. This linear dependence defines an explicit transfer
matrix between experimentally accessible control parameters and the
coefficients of the effective Hamiltonian, turning the synthesis of a
desired interaction pattern into a constrained linear control problem.

Building on this linear control map, we develop a target-aware synthesis
procedure for selecting a finite set of physically admissible
off-resonant tones and determining their non-negative drive intensities
for a prescribed interaction. Rather than optimizing the transfer matrix
for arbitrary interaction directions, the procedure selects tones
according to their ability to synthesize the chosen target channel while
respecting the detuning and drive-strength constraints required by the
dispersive approximation. We refer to this procedure as
\emph{dispersive interaction via analytical linear inversion} (DIAL).

We then test the resulting controls directly under the full driven
mediator-register dynamics. For register sizes \(n=2,3,4,5\), we
generate ensembles of independently sampled configuration-dependent
transition spectra and synthesize the corresponding highest-order
\(Z_1\cdots Z_n\) interaction, which, for an interaction phase of magnitude \(\pi/4\), transforms \(\lvert+\rangle^{\otimes n}\) into a state locally equivalent to an \(n\)-qubit GHZ state. The comparison between the analytical
dispersive construction and the full dynamics quantifies the implemented
register interaction, the maximum transient mediator excitation, and
the interaction time relative to the strongest applied drive. These
calculations test whether the linear control construction continues to
produce accurate, excitation-suppressed interactions when propagated
under the dynamics from which the dispersive description is obtained
only approximately.

The remainder of this article is organized as follows.
Section~\ref{sec:general_dispersive_dynamics} introduces the driven
register-mediator model, derives its dispersive reduction, and obtains
the linear transfer-map formulation of Pauli-\(Z\) interaction rates.
Section~\ref{sec:phase_rate_optimization} develops the target-aware DIAL
procedure for selecting physically admissible tones and determining
their intensities for a prescribed interaction.
Section~\ref{sec:numerical-demonstrations} tests the resulting controls
under the full driven dynamics across the register sizes considered.
The Appendices provide detailed analytical derivations and numerical
implementation, examine the relation between spectral resolution and
interaction time, and quantify departures from the
configuration-preserving aligned regime.
Section~\ref{sec:summary} summarizes the results and discusses
extensions beyond this regime.

\section{Dispersive interaction via analytical linear inversion}

\subsection{Dispersively driven mediator-register dynamics}
\label{sec:general_dispersive_dynamics}

We consider a quantum register whose computational basis states
$\{\ket{\alpha}\}$ shift the transition frequency of an auxiliary
mediator. The mediator may possess many internal degrees of freedom, but
we assume that a single addressable transition has been isolated and
projected onto an effective two-level subspace with states
$\ket{g}$ and $\ket{e}$. The mediator may therefore correspond to an electronic spin,
a collective mode, or another controllable degree of freedom,
while the register degrees of freedom are assumed to be
long-lived compared to the mediator.

The purpose of this subsection is to establish a generic description of
dispersive mediator dynamics that is independent of the microscopic
realization. Rather than starting from a particular spin Hamiltonian, we
identify the minimal ingredients required for the control framework
developed in this work: an addressable mediator transition whose
frequency depends on the configuration of a long-lived quantum register,
as illustrated in Fig.~\ref{fig:configuration_dependent_mediator_transition}.

\begin{figure*}[t]
    \centering
    \includegraphics[width=\textwidth]
    {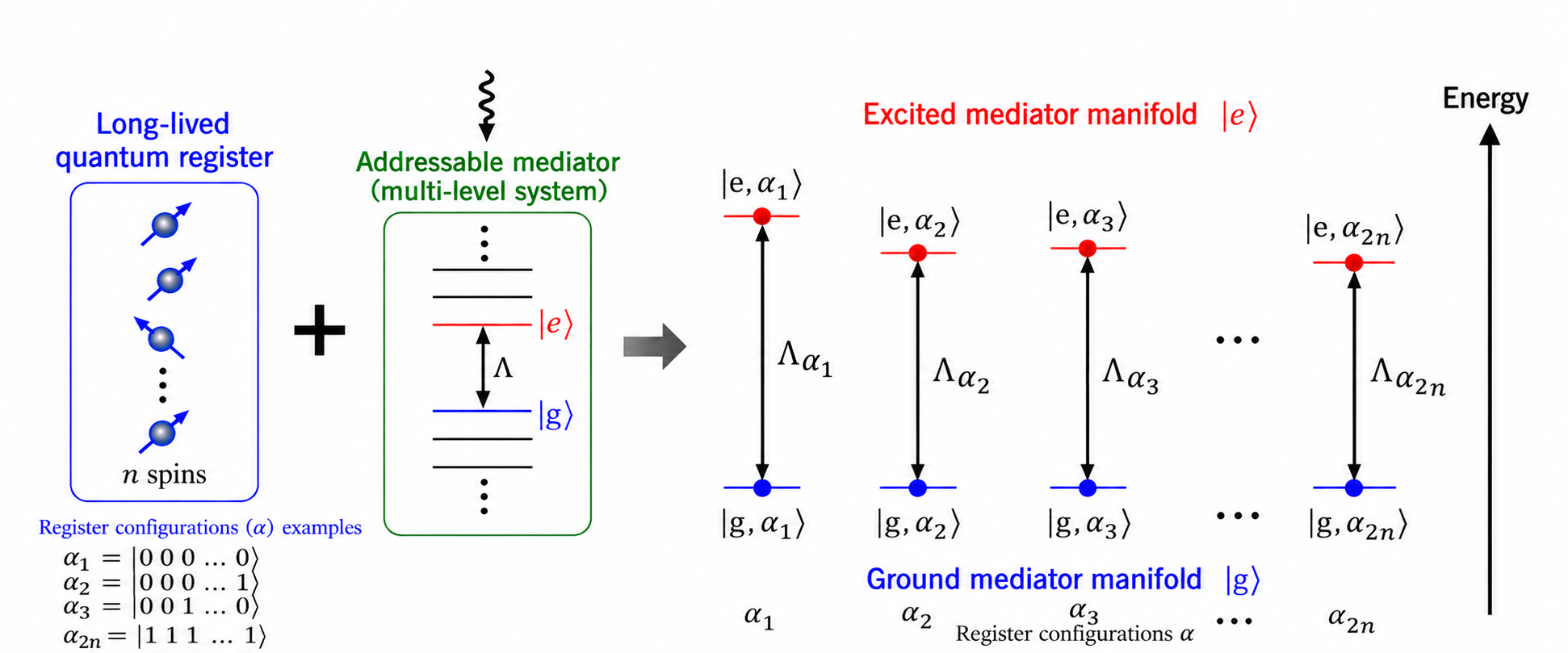}
    \caption{
    Schematic of a register-dependent mediator transition.
    The register configurations \(\alpha\) imprint themselves
    on the mediator transition, producing configuration-dependent
    transition frequencies \(\Lambda_\alpha\) and thereby enabling
    conditional manipulation through the mediator spectrum.
    }
    \label{fig:configuration_dependent_mediator_transition}
\end{figure*}

Once these ingredients are identified, the subsequent derivation is
completely general and applies irrespective of the physical mechanism
responsible for the configuration-dependent transition frequencies.

The derivation proceeds in three steps. We first formulate the projected
lab-frame Hamiltonian for the selected mediator transition. We then move
to the interaction picture and apply the rotating-wave approximation,
retaining the possibility that the register eigenbases associated with
the two mediator states differ. For the control construction, we
specialize to the configuration-preserving regime introduced in the
Introduction and subsequently eliminate the virtually excited mediator
in the dispersive limit. The resulting diagonal effective Hamiltonian
forms the starting point for the Pauli-$Z$ decomposition and linear
control map developed below.

\begin{proposition}[Projected lab-frame Hamiltonian]
\label{prop:general_lab_frame}

Let \(\{\ket{g},\ket{e}\}\) be the selected mediator transition and define
the projector onto this effective two-level subspace as
\[
P_{\mathrm{med}}
=
\ket{g}\!\bra{g}
+
\ket{e}\!\bra{e}.
\]
We restrict the full mediator-register system to
\(P_{\mathrm{med}}\otimes \mathds{1}_{\mathrm{reg}}\). Within this
projected mediator subspace, define
\[
\sigma_x
=
\ket{e}\!\bra{g}
+
\ket{g}\!\bra{e}.
\]
The resulting static Hamiltonian can be written as
\begin{equation}
H_0
=
\ket{g}\!\bra{g}\otimes H_g
+
\ket{e}\!\bra{e}\otimes H_e ,
\label{eq:projected_static_hamiltonian}
\end{equation}
where \(H_g\) and \(H_e\) are Hermitian operators acting on the register.

A multi-tone transverse drive with tone frequencies \(\omega_k\),
phases \(\varphi_k\), and physical Rabi-frequency envelopes
\(\Omega_k(t)\) is written as
\begin{equation}
H_{\mathrm{drive}}(t)
=
\sum_k
\Omega_k(t)
\cos(\omega_k t+\varphi_k)\,
\sigma_x
\otimes
\mathds{1}_{\mathrm{reg}} .
\label{eq:lab_frame_drive}
\end{equation}
The terms in Eq.~\eqref{eq:lab_frame_drive} represent independently applied physical drive tones.
The total projected lab-frame Hamiltonian is
\begin{equation}
H_{\mathrm{lab}}(t)
=
H_0
+
H_{\mathrm{drive}}(t).
\label{eq:projected_lab_frame_hamiltonian}
\end{equation}
\end{proposition}

Passing to the interaction picture generated by the static Hamiltonian \(H_0\) requires resolving the spectra of the
two register Hamiltonians. Their eigenbases need not coincide exactly,
so a mediator transition may in general be accompanied by a change of
register configuration. We therefore allow the register eigenbasis associated with
\(\lvert e\rangle\) to differ from that associated with
\(\lvert g\rangle\), as formalized below.

\begin{proposition}[Rotating-wave Hamiltonian]
\label{prop:general_rwa}

Let
\begin{equation}
\begin{aligned}
H_g\ket{\alpha_g}
&=
E_g(\alpha)\ket{\alpha_g},
\\
H_e(\eta)\ket{\beta_e(\eta)}
&=
E_e(\beta,\eta)\ket{\beta_e(\eta)},
\end{aligned}
\end{equation}
where \(\alpha\) and \(\beta\) label register eigenstates associated with
the mediator states \(\ket{g}\) and \(\ket{e}\), respectively. The
parameter \(\eta\) describes the relative mismatch between these
eigenbases, with \(\eta=0\) denoting the aligned case in which
\(H_e(0)\) shares the eigenbasis of \(H_g\). In this limit, the two bases may be labeled consistently such that
\(\lvert\beta_e(0)\rangle=\lvert\alpha_g\rangle\) for
\(\beta=\alpha\).
Define the overlap matrix
\(
M_{\beta\alpha}(\eta)
=
\braket{\beta_e(\eta)}{\alpha_g},
\)
the transition frequencies
\(
\Lambda_{\beta\alpha}(\eta)
=
E_e(\beta,\eta)-E_g(\alpha),
\)
and the corresponding detunings
\begin{equation}
\Delta_{k,\beta\alpha}(\eta)
=
\omega_k-\Lambda_{\beta\alpha}(\eta).
\end{equation}

Then, the interaction-picture Hamiltonian under the rotating-wave
approximation is
\begin{equation}
\small
\begin{aligned}
&H_I^{(\mathrm{RWA})}(t;\eta)
\\
&=
\sum_{k,\alpha,\beta}
\left[
\frac{\Omega_k(t)}{2}
M_{\beta\alpha}(\eta)
e^{-i\Delta_{k,\beta\alpha}(\eta)t+i\varphi_k}
\ket{e,\beta_e(\eta)}\!\bra{g,\alpha_g}
+
\mathrm{h.c.}
\right]
\end{aligned}
\label{eq:general_rwa_hamiltonian}
\end{equation}

In the aligned limit,
\(
M_{\beta\alpha}(0)
=
\delta_{\beta\alpha},
\)
so that only configuration-preserving mediator transitions remain.

\end{proposition}

\begin{proof}[Derivation]
See Appendix~\ref{app:proof_rwa_hamiltonian}. The proof derives the
general rotating-wave Hamiltonian for the family $H_e(\eta)$ and shows
that its aligned member $(\eta=0)$ reduces to the
configuration-preserving form used throughout the dispersive analysis.
\end{proof}

For the dispersive control construction, we consider the
configuration-preserving regime introduced above, a common operating
regime that arises naturally in many mediator-coupled systems or can be
approached in a high-field or secular limit.
Depending on the platform, this regime can arise from transition
selection rules, from the structure of the effective dispersive
interaction, or from an appropriate controlled operating regime
\cite{blais2021circuit,sorensen2000entanglement,
porras2004effective,aspelmeyer2014cavity,tabuchi2015coherent,
lachancequirion2019hybrid}.
In the present representation this corresponds to
\begin{equation}
M_{\beta\alpha}(0)=\delta_{\beta\alpha},
\end{equation}
so that mediator excitation preserves the register configuration and
register-dependent transition frequencies generate conditional energy
shifts without inducing register-state mixing. In this limit, only transitions with \(\beta=\alpha\) contribute.
We therefore drop \(\beta\) and use \(\alpha\) as the common register-state label for both mediator manifolds.
We denote the corresponding configuration-dependent transition frequency
by \(\Lambda_\alpha\), and define
\[
\Delta_{k\alpha}
\equiv
\omega_k-\Lambda_\alpha.
\]
Residual deviations from this limit are examined separately in the numerical analysis.

We now impose the dispersive condition anticipated above. In this regime, the applied tones satisfy
\begin{equation}
|\Omega_k(t)|
\ll
|\Delta_{k\alpha}|
\qquad
\forall\,k,\alpha .
\label{eq:general_dispersive_condition}
\end{equation}
Under this condition, real
mediator population transfer is suppressed and the mediator is only
virtually excited. Eliminating these virtual excitations yields the
following effective Hamiltonian.

\begin{proposition}[Dispersive Stark-shift Hamiltonian]
\label{prop:general_dispersive}

Assume the configuration-preserving limit of
Proposition~\ref{prop:general_rwa}, and suppose all drive tones
are far detuned from the resulting configuration-dependent mediator
transitions. With
the lab-frame convention of Eq.~\eqref{eq:lab_frame_drive}, the
drive-induced shift of the mediator transition frequency for register
configuration $\alpha$ is, to second order in the drive amplitudes,
\begin{equation}
\delta\Lambda_{\alpha}^{(\mathrm{Stark})}(t)
=
\sum_k
\frac{|\Omega_k(t)|^2}
{2\left[\omega_k-\Lambda_\alpha\right]},
\label{eq:explicit_stark_shift}
\end{equation}
provided that the beat frequencies between distinct tones are large
compared with the characteristic rates of the induced dispersive
dynamics and the inverse control timescale used for coarse graining, so
that cross terms between different tone frequencies average to zero.

The corresponding effective Hamiltonian is
\begin{equation}
H_{\mathrm{eff}}(t)
=
\frac12
\sigma_z
\otimes
\sum_\alpha
\delta\Lambda_{\alpha}^{(\mathrm{Stark})}(t)
P_\alpha .
\label{eq:general_dispersive_heff}
\end{equation}
The sum over $\alpha$ expresses that each register configuration acquires
a different ac Stark shift, while the factor $\sigma_z/2$ assigns
opposite energy shifts to the two mediator states.
\end{proposition}

\begin{proof}[Derivation]
See Appendix~\ref{app:proof_dispersive_effective_hamiltonian}. The proof
derives the effective Hamiltonian by adiabatic elimination,
equivalently a second-order Schrieffer-Wolff transformation, and
obtains the explicit form of the Stark shifts together with the
conditions under which contributions from multiple drive tones add
linearly.
\end{proof}

In the configuration-preserving dispersive regime, the driven mediator
therefore generates a diagonal, register-state-dependent Stark-shift
Hamiltonian. This diagonal structure is the key ingredient for the
control construction: it allows the effective interaction to be expanded
in Pauli-$Z$ channels whose rates can be related directly to the applied
drive powers.

\subsection{Pauli-\texorpdfstring{$Z$}{Z} representation}
\label{sec:pauli_z_representation}

Proposition~\ref{prop:general_dispersive} establishes that dispersive
driving generates an effective Hamiltonian in terms of configuration-dependent Stark shifts. While this representation is natural from the perspective of the underlying physics, it does not directly reveal the structure of the resulting multi-qubit interactions.
For Hamiltonian engineering it is more convenient to express the effective interaction in an operator basis whose coefficients correspond to physically interpretable interaction channels.

We therefore consider a register consisting of $n$ effective qubits with
computational basis
\begin{equation}
\ket{\alpha}
=
\ket{\alpha_1,\ldots,\alpha_n},
\qquad
\alpha_i\in\{0,1\}.
\end{equation}
For every subset
$S\subseteq\{1,\ldots,n\}$,
we define the corresponding Pauli-$Z$ string
\begin{equation}
Z_S
=
\prod_{i\in S}Z_i,
\label{eq:pauli_string}
\end{equation}
with the convention
\(
Z_\emptyset=\mathds{1}.
\)
Acting on a computational basis state gives
\begin{equation}
Z_S\ket{\alpha}
=
\left(
\prod_{i\in S}
(-1)^{\alpha_i}
\right)
\ket{\alpha}.
\label{eq:pauli_eigenvalues}
\end{equation}

The configuration-dependent Stark shifts derived in
Proposition~\ref{prop:general_dispersive} define the diagonal operator
\begin{equation}
D(t)
=
\sum_\alpha
\delta\Lambda_{\alpha}^{(\mathrm{Stark})}(t)
P_\alpha.
\label{eq:stark_operator}
\end{equation}
Since the Pauli-$Z$ strings form an orthogonal basis for all diagonal
operators on the register, $D(t)$ admits the unique expansion
\begin{equation}
D(t)
=
\sum_{S\subseteq\{1,\ldots,n\}}
K_S(t)\,
Z_S,
\label{eq:pauli_expansion}
\end{equation}
where the expansion coefficients are obtained by Hilbert-Schmidt
projection,
\begin{equation}
K_S(t)
=
2^{-n}
\operatorname{Tr}
\!\left[
D(t)Z_S
\right].
\label{eq:projection}
\end{equation}
Evaluating the trace in the computational basis gives
\begin{equation}
K_S(t)
=
2^{-n}
\sum_\alpha
\delta\Lambda_{\alpha}^{(\mathrm{Stark})}(t)
\prod_{i\in S}
(-1)^{\alpha_i}.
\label{eq:projection_explicit}
\end{equation}

Substituting Eq.~\eqref{eq:pauli_expansion} into the effective
Hamiltonian of Proposition~\ref{prop:general_dispersive} yields
\begin{equation}
H_{\mathrm{eff}}(t)
=
\frac12
\sigma_z
\otimes
\sum_{S\subseteq\{1,\ldots,n\}}
K_S(t)\,
Z_S.
\label{eq:heff_pauli}
\end{equation}

Equation~\eqref{eq:heff_pauli} provides an interpretable decomposition of
the dispersive interaction into mediator-register coupling channels. The
coefficient associated with the empty subset,
$K_\emptyset$, corresponds to a global ac Stark shift of the mediator.
Single-qubit terms describe effective local fields conditioned on the
mediator state, two-qubit strings correspond to Ising-type interactions,
and higher-order strings represent genuine multi-qubit couplings. The
coefficients $\{K_S(t)\}$ therefore quantify the instantaneous strength
of every interaction channel generated by the dispersive drive.

\begin{definition}[Pauli-\(Z\) interaction channels and rates]
\label{def:pauli_z_interaction_channels}
For each nonempty subset \(S\subseteq\{1,\ldots,n\}\), the Pauli-\(Z\)
string \(Z_S\) appearing in Eq.~\eqref{eq:heff_pauli} defines a diagonal
interaction channel, with interaction rate \(K_S(t)\) given explicitly
by Eq.~\eqref{eq:projection_explicit}. The coefficient \(K_{\emptyset}\) multiplies
the register identity and therefore generates only a local mediator
\(Z\) rotation rather than a nontrivial register interaction.
\end{definition}


The microscopic details of the mediator and register determine the
configuration-dependent Stark shifts, while the Pauli-\(Z\)
decomposition identifies the effective multi-qubit interaction channels
that they generate. In the following subsection we show that, for fixed
tone frequencies and constant-amplitude multi-tone driving, the strengths
of these interactions depend linearly on the applied drive powers,
establishing a direct correspondence between the drive parameters and the
resulting effective interactions.

\subsection{Linear control map}
\label{sec:linear_hamiltonian_engineering}

The effective Hamiltonian of
Proposition~\ref{prop:general_dispersive} shows that each dispersive
drive tone generates a configuration-dependent ac Stark shift.
After the Pauli-\(Z\) decomposition of the previous subsection, these
shifts determine the strengths of the different Pauli-\(Z\) interaction
channels.

This motivates formulating the control problem directly in terms of these
interaction channels. We use \emph{dispersive interaction engineering} to
refer to selecting off-resonant drive-tone parameters so that the resulting
effective Hamiltonian realizes prescribed strengths \(K_S\) of selected
nontrivial Pauli-\(Z\) interaction channels while suppressing real mediator
excitation.

For the linear control construction, we specialize the general drive of
Eq.~\eqref{eq:lab_frame_drive} to constant-amplitude multitone controls.

\begin{definition}[Constant-amplitude multitone control]
\label{def:constant_amplitude_multitone_control}
A constant-amplitude multitone control is the specialization
\(\Omega_k(t)=\Omega_k\) of the drive in
Eq.~\eqref{eq:lab_frame_drive}, with fixed tone frequencies
\(\omega_k\), phases \(\varphi_k\), and physical Rabi-frequency
amplitudes \(\Omega_k\), for \(k=1,\ldots,M\).
For each tone, we define the non-negative intensity
\begin{equation}
x_k
=
\frac{|\Omega_k|^2}{2}.
\label{eq:intensity_variable}
\end{equation}
\end{definition}

For such controls, the interaction rates are time independent and
independently detuned tones contribute additively in the dispersive regime.
In terms of the intensity variables of
Definition~\ref{def:constant_amplitude_multitone_control}, the multitone
Stark shift is
\begin{equation}
\delta\Lambda_{\alpha}^{(\mathrm{Stark})}
=
\sum_k
\frac{x_k}{\omega_k-\Lambda_\alpha}.
\label{eq:stark_shift_linear_x}
\end{equation}

\begin{theorem}[Linear dispersive control map]
\label{thm:linear_hamiltonian_engineering}

Consider the Pauli-\(Z\) interaction channels of
Definition~\ref{def:pauli_z_interaction_channels} under a
constant-amplitude multitone control of
Definition~\ref{def:constant_amplitude_multitone_control},
with fixed tone frequencies in the dispersive regime of
Proposition~\ref{prop:general_dispersive}. The effective Hamiltonian has
the form
\begin{equation}
H_{\mathrm{eff}}
=
\frac{1}{2}
\sigma_z
\otimes
\sum_S
K_S Z_S .
\label{eq:linear_hamiltonian_engineering}
\end{equation}

For the control map, fix an ordering of the nonempty subsets of
\(\{1,\ldots,n\}\),
\begin{equation}
\mathcal S_\ast
=
\{S_1,\ldots,S_{2^n-1}\}.
\label{eq:control_channel_ordering}
\end{equation}
With this ordering, define
\begin{equation}
\begin{aligned}
\mathbf K
&=
\bigl(
K_{S_1},\ldots,K_{S_{2^n-1}}
\bigr)^T,
\\
\mathbf x
&=
(x_1,\ldots,x_M)^T,
\\
x_k
&=
\frac{|\Omega_k|^2}{2}.
\end{aligned}
\label{eq:control_vectors}
\end{equation}
Thus
\(
\mathbf K\in\mathbb R^{2^n-1}
\)
and
\(
\mathbf x\in\mathbb R_{\ge0}^{M}.
\)
The identity coefficient \(K_{\emptyset}\) is omitted from
\(\mathbf K\), since it does not constitute a register-interaction
channel.

These vectors are related by the linear control map
\begin{equation}
\boxed{
\mathbf K
=
G\mathbf x
}.
\label{eq:linear_hamiltonian_engineering_system}
\end{equation}
Here \(G\in\mathbb R^{(2^n-1)\times M}\) is the static transfer matrix
fixed by the tone frequencies \(\{\omega_k\}\) and the
register-resolved transition spectrum, with entries
\begin{equation}
G_{j,k}
=
\frac{1}{2^n}
\sum_{\alpha}
\frac{1}{\omega_k-\Lambda_\alpha}
\prod\nolimits_{i\in S_j}
(-1)^{\alpha_i},
\qquad
j=1,\ldots,2^n-1 .
\label{eq:transfer_matrix}
\end{equation}

For each interaction channel, the corresponding interaction-rate
coefficient is therefore
\begin{equation}
K_{S_j}
=
\sum_{k=1}^{M}
G_{j,k}x_k,
\qquad
j=1,\ldots,2^n-1 .
\label{eq:linear_hamiltonian_engineering_components}
\end{equation}

Thus \(G\) maps the non-negative tone-intensity vector \(\mathbf x\)
to the resulting Pauli-\(Z\) interaction-rate coefficients
\(\mathbf K\). Once the tone frequencies have been chosen, dispersive
Hamiltonian engineering therefore reduces to this linear control map.

\end{theorem}

\begin{proof}
For the constant-amplitude multitone controls of
Definition~\ref{def:constant_amplitude_multitone_control}, the
interaction-rate coefficients given by
Eq.~\eqref{eq:projection_explicit} are time independent. Substituting the multitone Stark shift of
Eq.~\eqref{eq:stark_shift_linear_x} into
Eq.~\eqref{eq:projection_explicit} gives
\begin{equation}
K_S
=
2^{-n}
\sum_\alpha
\left[
\sum_k
\frac{x_k}
{\omega_k-\Lambda_\alpha}
\right]
\prod_{i\in S}
(-1)^{\alpha_i}.
\end{equation}

Interchanging the sums gives
\begin{equation}
K_S
=
\sum_k
x_k
\left[
\frac{1}{2^n}
\sum_\alpha
\frac{1}{\omega_k-\Lambda_\alpha}
\prod\nolimits_{i\in S}
(-1)^{\alpha_i}
\right].
\end{equation}

For the fixed ordering of interaction channels given in
Eq.~\eqref{eq:control_channel_ordering}, the \(j\)th channel is \(S=S_j\),
and the quantity in brackets is precisely the transfer-matrix element
\(G_{j,k}\) defined in Eq.~\eqref{eq:transfer_matrix}. Therefore,
\begin{equation}
K_{S_j}
=
\sum_{k=1}^{M}
G_{j,k}x_k,
\qquad
j=1,\ldots,2^n-1 .
\end{equation}

Collecting these component relations into the vectors defined in
Eq.~\eqref{eq:control_vectors} gives
\begin{equation}
\mathbf K
=
G\mathbf x ,
\end{equation}
which proves the linear control map of
Eq.~\eqref{eq:linear_hamiltonian_engineering_system}.
\end{proof}

Theorem~\ref{thm:linear_hamiltonian_engineering} is the central
structural result of the framework. It provides an explicit linear map
from experimentally controlled tone intensities to the resulting
multi-qubit interaction rates. Once the carrier frequencies are fixed,
the transfer matrix is determined by the register-resolved mediator
spectrum, so that the effect of each experimentally controlled tone intensity on
every engineered interaction channel is known explicitly.

The coefficient \(K_{\emptyset}\) is excluded from the
control-map formulation because it generates only a local mediator
\(Z\) rotation, which can be corrected virtually and does not influence
the interaction induced on the register. 
Once a control has been specified, however, its value remains directly
determined by the same dispersive shifts,
\begin{equation}
    K_{\emptyset}
    =
    \sum_{k=1}^{M}
    x_k
    \left[
        \frac{1}{2^n}
        \sum_{\alpha}
        \frac{1}{\omega_k-\Lambda_\alpha}
    \right].
\label{eq:identity_channel_coefficient}
\end{equation}

The components of \(\mathbf K\) determine the engineered interaction
rates, while \(K_{\emptyset}\) specifies the separately tracked local
mediator contribution. Together, these coefficients determine the
effective evolution generated by the synthesized control.

\begin{corollary}[Effective evolution from synthesized interaction rates]
\label{cor:effective_evolution}
For an interaction-rate vector \(\mathbf K\) generated by the linear
control map, together with the corresponding coefficient
\(K_{\emptyset}\) determined from
Eq.~\eqref{eq:identity_channel_coefficient}, the effective Hamiltonian is
\begin{equation}
    H_{\mathrm{eff}}
    =
    \frac{1}{2}\sigma_z
    \otimes
    \left(
        K_{\emptyset}\mathds{1}
        +
        \sum_{S\neq\emptyset} K_S Z_S
    \right).
\end{equation}
Since all terms commute, the corresponding propagator after an
interaction time \(T\) is
\begin{equation}
    U_{\mathrm{eff}}(T)
    =
    e^{-i K_{\emptyset}T\sigma_z/2}
    \exp\!\left[
        -\frac{iT}{2}
        \sigma_z
        \otimes
        \sum_{S\neq\emptyset} K_S Z_S
    \right].
\end{equation}
The first factor is a local mediator \(Z\) rotation,
\(R_z(\phi_{\mathrm{med}})=e^{-i\phi_{\mathrm{med}}\sigma_z/2}\),
with accumulated rotation angle
\begin{equation}
    \phi_{\mathrm{med}}(T)
    =
    K_{\emptyset}T .
\label{eq:mediator_phase}
\end{equation}
Since this rotation acts only on the mediator, it can be corrected
virtually by a mediator-frame update and does not modify the engineered
register interaction.

For a mediator prepared in a
\(\sigma_z\) eigenstate with eigenvalue
\(s_{\mathrm{med}}=\pm1\), the induced register evolution is therefore
\begin{equation}
    U_{\mathrm{reg}}(T)
    =
    \exp\!\left[
        -\frac{iT}{2}
        s_{\mathrm{med}}
        \sum_{S\neq\emptyset} K_S Z_S
    \right].
\end{equation}
Hence each synthesized coefficient \(K_S\) is an interaction rate,
and the corresponding accumulated Pauli-string phase is
\begin{equation}
    \Phi_S(T)
    =
    \frac{1}{2}s_{\mathrm{med}}K_S T .
\end{equation}
\end{corollary}

Thus the linear control map determines the interaction content and rates
of the effective Hamiltonian, whereas the accumulated interaction
phases are set by the evolution time. Once an interaction-rate vector
has been synthesized, its accumulated interaction phase is therefore
selected through the interaction time. Before turning to target-specific
synthesis, we first state the physical constraints that define the
admissible domain of this linear control map.

\subsection{Physical constraints on the linear control map}
\label{sec:linear_control_constraints}

For fixed tone frequencies, the transfer matrix of
Theorem~\ref{thm:linear_hamiltonian_engineering} defines a
target-independent linear map. Physical admissibility restricts the
domain of this map. Since
\(x_k=|\Omega_k|^2/2\), the tone intensities satisfy
\(x_k\ge0\). The accessible nontrivial interaction-rate vectors
therefore form the positive cone
\begin{equation}
\mathcal C_G
=
\left\{
G\mathbf x
\,:\,
\mathbf x\in\mathbb R_{\ge0}^{M}
\right\}.
\label{eq:accessible_positive_cone}
\end{equation}

The dispersive approximation further constrains the drive amplitude
at each tone frequency. For tone \(k\), define the nearest
register-resolved transition detuning
\begin{equation}
d_k
=
\min_\alpha
|\omega_k-\Lambda_\alpha|.
\label{eq:theory_nearest_detuning}
\end{equation}
Dispersive operation requires the drive to remain small compared with
this detuning. We make this constraint explicit by prescribing an
admissible maximum drive-to-detuning ratio,
\begin{equation}
\max_k
\frac{|\Omega_k|}{d_k}
\le
r_{\mathrm{disp}},
\label{eq:optimization_constraints}
\end{equation}
where \(r_{\mathrm{disp}}\) is chosen sufficiently small for the
intended perturbative regime. Consequently, once the tone frequencies
and the allowed dispersive ratio have been specified, the admissible
drive amplitude of each tone is bounded by
\begin{equation}
|\Omega_k|
\le
r_{\mathrm{disp}}d_k.
\label{eq:dispersive_amplitude_bound}
\end{equation}
Using \(x_k=|\Omega_k|^2/2\), this becomes the corresponding intensity
constraint
\begin{equation}
0
\le
x_k
\le
\frac{r_{\mathrm{disp}}^2d_k^2}{2}.
\label{eq:dispersive_intensity_bound}
\end{equation}
Thus the allowed physical intensity is not chosen independently of the
spectrum: for a prescribed degree of dispersivity, the detuning of each
tone determines how strongly that tone may be driven.

The ratio \(r_{\mathrm{disp}}\) controls the perturbative drive strength.
In particular, mediator excitation in the dispersive regime scales as
\(O[(\Omega_k/\Delta_{k\alpha})^2]\), together with higher-order
corrections to the effective dynamics. A prescribed value of
\(r_{\mathrm{disp}}\) therefore defines an admissible perturbative
control regime, but does not by itself guarantee a particular gate
fidelity. We test mediator excitation and interaction-phase accuracy
separately under the full dynamics in
Sec.~\ref{sec:numerical-demonstrations}.

In addition, the cross-tone averaging used in
Proposition~\ref{prop:general_dispersive} requires the beat frequencies
between distinct tones to be large compared with the characteristic
rates of the induced dispersive dynamics and the inverse control
timescale used for coarse graining. Optional hardware limitations may
further impose bounds of the form
\begin{equation}
0
\le
x_k
\le
x_k^{\mathrm{max}} .
\end{equation}
These conditions restrict the physically admissible domain of the
linear control map but do not themselves specify a target interaction.

The physical constraints above therefore define the target-independent
admissible domain of the linear control map. They also leave an overall
control-strength degree of freedom: once an interaction has been
synthesized, the same relative interaction content can be generated at
different drive strengths within the admissible dispersive regime. We
next show that this freedom produces a direct tradeoff between
dispersivity and gate time.

\subsection{Dispersivity-gate-time tradeoff}
\label{sec:dispersivity_gate_time_tradeoff}

The physical constraints above also expose a direct tradeoff between
dispersivity and gate time. Consider an interaction synthesized through
the linear control map of
Theorem~\ref{thm:linear_hamiltonian_engineering}, with tone-intensity
vector \(\mathbf{x}\) and corresponding interaction-rate vector
\begin{equation}
    \mathbf{K}
    =
    G\mathbf{x}.
\label{eq:tradeoff_reference_control}
\end{equation}
For a target interaction \(Z_{S^\ast}\),
Corollary~\ref{cor:effective_evolution} gives
\begin{equation}
    \left|
        \Phi_{S^\ast}^{\mathrm{tar}}
    \right|
    =
    \frac{1}{2}
    \left|
        K_{S^\ast}
    \right|
    T_{\mathrm{gate}},
\label{eq:tradeoff_reference_gate_time}
\end{equation}
and hence
\begin{equation}
    T_{\mathrm{gate}}
    =
    \frac{
        2\left|
            \Phi_{S^\ast}^{\mathrm{tar}}
        \right|
    }{
        \left|
            K_{S^\ast}
        \right|
    }.
\label{eq:tradeoff_reference_gate_time_explicit}
\end{equation}

Once the relative interaction content has been synthesized, its overall
strength remains adjustable without changing that content. For fixed
tone frequencies, multiply all tone intensities by a common positive
scale factor \(\zeta\),
\begin{equation}
    \mathbf{x}(\zeta)
    =
    \zeta\mathbf{x}.
\label{eq:tradeoff_intensity_scaling}
\end{equation}
Since the transfer matrix is unchanged,
\begin{equation}
    \mathbf{K}(\zeta)
    =
    G\mathbf{x}(\zeta)
    =
    \zeta\mathbf{K}.
\label{eq:tradeoff_rate_scaling}
\end{equation}
Thus all interaction rates acquire the same overall factor while their
relative values remain unchanged.

Because \(x_k=|\Omega_k|^2/2\), the corresponding drive amplitudes obey
\begin{equation}
    |\Omega_k(\zeta)|
    =
    \sqrt{\zeta}\,|\Omega_k|.
\label{eq:tradeoff_amplitude_scaling}
\end{equation}
The tone frequencies and nearest-transition detunings \(d_k\) remain
fixed. Therefore, if the initial control has dispersive ratio
\begin{equation}
    r_{\mathrm{disp}}
    =
    \max_k\frac{|\Omega_k|}{d_k},
\end{equation}
the rescaled control has
\begin{equation}
    r_{\mathrm{disp}}(\zeta)
    =
    \sqrt{\zeta}\,r_{\mathrm{disp}}.
\label{eq:tradeoff_rdisp_scaling}
\end{equation}
At the same time, Eq.~\eqref{eq:tradeoff_reference_gate_time_explicit}
and the rate scaling in
Eq.~\eqref{eq:tradeoff_rate_scaling} give
\begin{equation}
    T_{\mathrm{gate}}(\zeta)
    =
    \frac{T_{\mathrm{gate}}}{\zeta}.
\label{eq:tradeoff_gate_time_scaling}
\end{equation}
Eliminating \(\zeta\) therefore gives
\begin{equation}
    T_{\mathrm{gate}}(\zeta)
    =
    T_{\mathrm{gate}}
    \left[
        \frac{r_{\mathrm{disp}}}
             {r_{\mathrm{disp}}(\zeta)}
    \right]^2 .
\label{eq:tradeoff_gate_time_relation}
\end{equation}

\begin{proposition}[Dispersivity-gate-time tradeoff]
\label{prop:dispersivity_gate_time_tradeoff}
For fixed tone frequencies and fixed relative tone intensities, changing
the overall control strength preserves the relative interaction content
while changing the interaction rates and gate time according to
\begin{equation}
    K_S
    \propto
    r_{\mathrm{disp}}^2,
    \qquad
    T_{\mathrm{gate}}
    \propto
    r_{\mathrm{disp}}^{-2}.
\label{eq:dispersivity_gate_time_summary}
\end{equation}
\end{proposition}

The dispersive ratio therefore does more than define the perturbative
domain of the effective description: it directly sets the speed at
which a synthesized interaction can be accumulated. Its physical meaning
can be made equally explicit by relating the same ratio to mediator
population.
As shown below, this turns \(r_{\mathrm{disp}}\) into a practical design
parameter: an acceptable mediator-population scale can be translated directly
into the corresponding interaction-time cost, and vice versa.

\begin{proposition}[Dispersivity and mediator excitation]
\label{prop:dispersivity_mediator_excitation}
For an isolated off-resonant mediator transition with
drive-to-detuning ratio \(r=|\Omega|/|\Delta|\), detuned Rabi dynamics
gives the maximum mediator population
\begin{equation}
    P_{\mathrm{med}}^{\max}
    =
    \frac{r^2}{1+r^2}.
\label{eq:detuned_rabi_population_max}
\end{equation}
Consequently, the dispersive constraint
\(|\Omega_k|/d_k\le r_{\mathrm{disp}}\) defines the corresponding
single-transition population scale
\begin{equation}
    P_{\mathrm{med}}^{\mathrm{disp}}
    =
    \frac{r_{\mathrm{disp}}^2}
         {1+r_{\mathrm{disp}}^2}
    \simeq
    r_{\mathrm{disp}}^2,
    \qquad
    r_{\mathrm{disp}}\ll1 .
\label{eq:dispersivity_population}
\end{equation}
\end{proposition}

The same parameter thus controls both the interaction speed and the
characteristic mediator population. Eliminating
\(r_{\mathrm{disp}}\) between
Propositions~\ref{prop:dispersivity_gate_time_tradeoff}
and~\ref{prop:dispersivity_mediator_excitation} gives their direct
tradeoff.

\begin{corollary}[Mediator-population-gate-time tradeoff]
\label{cor:mediator_population_gate_time_tradeoff}
\begin{equation}
    T_{\mathrm{gate}}
    \propto
    \frac{
        1-P_{\mathrm{med}}^{\mathrm{disp}}
    }{
        P_{\mathrm{med}}^{\mathrm{disp}}
    }
    \simeq
    \frac{1}{
        P_{\mathrm{med}}^{\mathrm{disp}}
    },
    \qquad
    P_{\mathrm{med}}^{\mathrm{disp}}\ll1 .
\label{eq:population_gate_time_tradeoff}
\end{equation}
\end{corollary}

Here \(P_{\mathrm{med}}^{\mathrm{disp}}\) characterizes the isolated
transition associated with the dispersive constraint, rather than the
actual transient mediator population generated by the complete
multitone drive. The latter is evaluated directly under the full driven
dynamics in Sec.~\ref{sec:numerical-demonstrations}. If an additional
hardware intensity bound becomes active, it can instead set the
available overall control strength; the relations above describe the
regime in which dispersivity provides the limiting scale.

\section{Target-Aware DIAL Algorithm}
\label{sec:phase_rate_optimization}

We now turn the linear dispersive control map into a practical procedure
for synthesizing a desired interaction. We refer to this procedure as the
\emph{Target-Aware DIAL algorithm}. For a chosen target Pauli string
\(Z_{S^\ast}\), the objective is to generate an interaction-rate vector
directed predominantly along the corresponding interaction channel,
\begin{equation}
    \mathbf K
    \simeq
    K_{S^\ast}\mathbf e_{S^\ast},
\end{equation}
where \(\mathbf e_{S^\ast}\) is the unit vector associated with
\(Z_{S^\ast}\). The target rate \(K_{S^\ast}\) is thereby retained while
the spectator rates are suppressed. Its magnitude sets the interaction
rate, whereas the accumulated target phase is set by the subsequent
interaction time, as established in
Corollary~\ref{cor:effective_evolution}.

The linear control map makes the action of the controls transparent. For any
fixed combination of tone frequencies, it maps the non-negative tone
intensities directly onto the resulting interaction-rate spectrum through
the associated transfer matrix. For a mediator fixed in a
\(\sigma_z\) eigenstate, these interaction rates correspond directly to
register phase rates up to a fixed factor. The tone frequencies themselves
determine the structure of the transfer matrix and thereby how selectively
such an interaction pattern can be synthesized. If a favorable choice of
tone frequencies is given, the available range of tone intensities can
therefore be used to generate a phase-rate spectrum with a strong target
channel and suppressed spectator channels. The subsequent interaction time
sets the accumulated interaction phases. This desired control situation is
illustrated schematically in
Fig.~\ref{fig:target_aware_control_synthesis}.

\begin{figure*}[t]
    \centering
    \includegraphics[width= \textwidth]
    {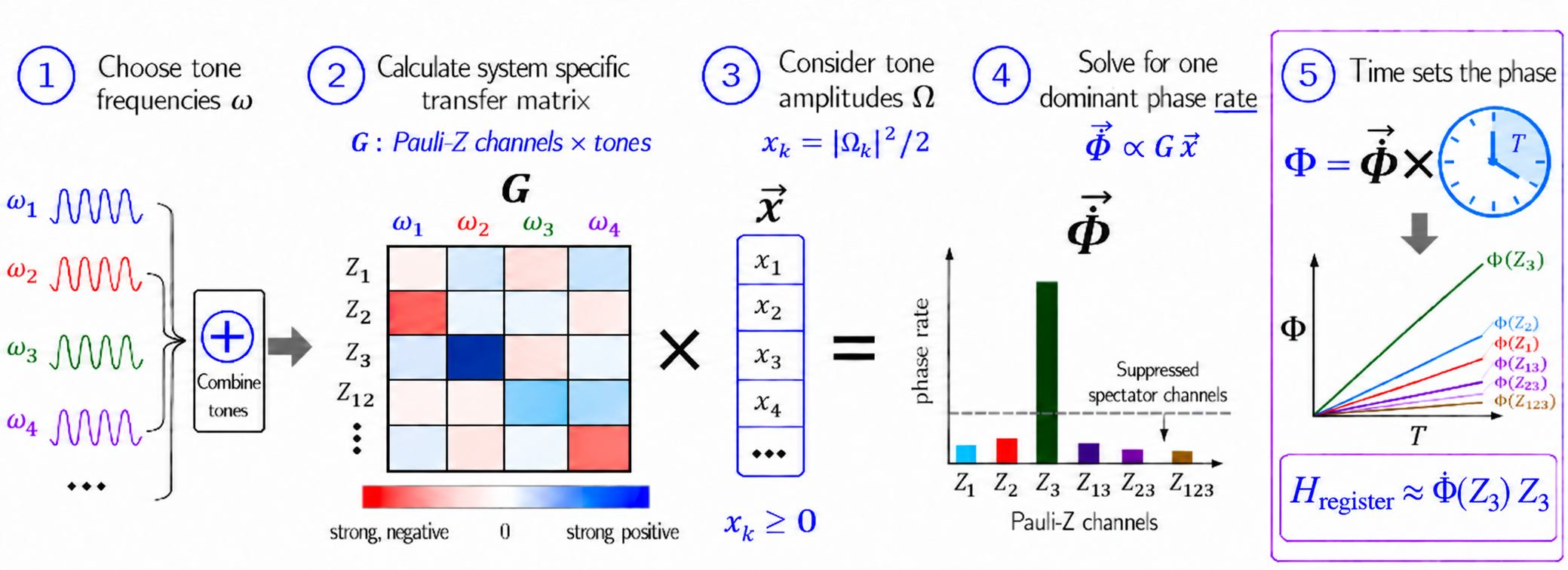}
    \caption{
    Schematic of the desired DIAL control situation for a favorable
    combination of tone frequencies.
    (1) A suitably chosen set of off-resonant tone frequencies
    \(\{\omega_k\}\) defines the available controls.
    (2) Together with the register-resolved transition spectrum, this tone
    set determines the transfer matrix \(G\), which encodes how the available
    controls contribute across the Pauli-\(Z\) interaction channels.
    (3) The non-negative drive-intensity vector
    \(\mathbf{x}=(x_1,\ldots,x_M)^{\mathsf T}\) weights the selected tones
    within the admissible intensity range.
    (4) For a favorable tone set, the available tone intensities can generate
    a phase-rate spectrum
    \(\dot{\boldsymbol{\Phi}}\propto G\mathbf{x}\)
    with a strong target channel and suppressed spectator channels.
    (5) The subsequent interaction time \(T\) sets the accumulated
    interaction phases through
    \(\boldsymbol{\Phi}(T)=\dot{\boldsymbol{\Phi}}T\).
    For the illustrated \(Z_3\) target, suppression of the spectator
    phase rates gives
    \(H_{\mathrm{register}}\approx\dot{\Phi}(Z_3)Z_3\).
    }
    \label{fig:target_aware_control_synthesis}
\end{figure*}
However, the linear control map alone
does not prescribe which combination of tone frequencies from the available
repertoire is most favorable for the desired interaction, for example in
terms of the highest attainable contrast between the target and spectator
interaction rates. The \emph{Target-Aware DIAL algorithm} addresses this
remaining selection problem. It operates on the core principles of searching
the admissible tone repertoire, greedily constructing a tone set, re-solving
the non-negative linear inversion for each trial addition, and retaining the
most favorable candidate at each step.

For a practical finite-tone construction, we discretize the admissible
frequency window into a set of candidate tones and build the control from
this finite pool. Tone selection is separated from intensity determination:
for each trial tone set, the optimal drive intensities are obtained
from a convex linear problem, avoiding a joint optimization over continuous
frequencies and intensities. Starting from an empty tone set, candidates
are added greedily: at each step, each admissible remaining frequency is
appended in turn, the drive intensities of the combined trial set are
solved anew through the linear control map, and the candidate giving the
best combined target-channel performance is retained. This procedure is repeated up to the system-size-dependent tone budget introduced below. Off-resonant, tone-spacing, and dispersive constraints
are imposed throughout so that all retained controls remain within the
regime of validity of the linear dispersive description.

We first construct the finite candidate-frequency pool. To make the same
procedure applicable to spectra with different bandwidths and different
numbers of register-resolved transitions, we express the frequency-resolution
constraints relative to the spectral structure of the supplied system.

For an \(n\)-spin register with ordered transition frequencies
\(\Lambda_{(1)}\le\cdots\le\Lambda_{(2^n)}\), we define
\begin{equation}
\bar{\Delta}
=
\frac{\Lambda_{\max}-\Lambda_{\min}}{2^n-1},
\qquad
\Delta_{\min}
=
\min_j
\left[
\Lambda_{(j+1)}-\Lambda_{(j)}
\right].
\label{eq:characteristic_transition_spacings}
\end{equation}
Here \(\bar{\Delta}\) is the mean adjacent spacing of the \(2^n\)
register-resolved transition frequencies over their full spectral width,
and \(\Delta_{\min}\) is their minimum adjacent spacing.

Rather than specifying absolute frequency thresholds, we set the
minimum allowed tone-to-transition detuning directly from the narrowest
spectral interval, while retaining a dimensionless tone-spacing
parameter \(\eta_\omega\),
\begin{equation}
d_{\min}
=
0.45\,\Delta_{\min},
\qquad
\delta\omega_{\min}
=
\eta_\omega\,\bar{\Delta},
\qquad
\eta_\omega=0.075.
\label{eq:normalized_frequency_resolution}
\end{equation}
Here \(d_{\min}\) is the minimum allowed detuning of a control tone from
any register-resolved transition and \(\delta\omega_{\min}\) is the
minimum spacing between two tones retained in the same control. The
choice \(d_{\min}=0.45\,\Delta_{\min}\) leaves \(10\%\) of even the
narrowest inter-transition interval available for tone placement.
The value \(\eta_\omega=0.075\) is used as a practical default rather than
a fundamental threshold. It prevents the algorithm from retaining nearly
degenerate tones, for which small beat frequencies can compromise the
independent-tone dispersive description, while still leaving a sufficiently
dense candidate set for interaction synthesis. The parameter may be increased
when greater tone separation is required by the physical platform or control
hardware.

Taking the center of the register-resolved transition-frequency spectrum,
\(\Lambda_c=(\Lambda_{\max}+\Lambda_{\min})/2\), as the frequency origin,
we write the transition frequencies relative to \(\Lambda_c\) and define
the frequency window
\begin{equation}
B=\max_\alpha |\Lambda_\alpha|,
\qquad
\mathcal W=[-1.35B,\,1.35B].
\label{eq:default_frequency_window}
\end{equation}
We discretize this frequency window uniformly into
\(N_{\mathrm{grid}}=6401\) frequencies. The \(35\%\) spectral margin is used
as a practical default to provide access to off-resonant tones on both sides
of the transition band without extending the search unnecessarily far from
the register-resolved spectrum. Likewise, \(N_{\mathrm{grid}}=6401\) is a
numerical resolution parameter rather than a fundamental choice; it provides
a dense candidate-frequency pool for the tone-selection step and may be
increased when finer frequency resolution is required.
Candidate frequencies lying within \(d_{\min}\) of any register-resolved transition are removed,
\begin{equation}
\Omega_{\mathrm{cand}}
=
\left\{
\omega\in\Omega_{\mathrm{grid}}
:
\min_\alpha
|\omega-\Lambda_\alpha|
\ge
d_{\min}
\right\}.
\label{eq:candidate_tone_pool}
\end{equation}
The discretization is therefore a practical restriction of the
continuous carrier-frequency choice rather than an additional assumption
of the analytical control map. The minimum spacing
\(\delta\omega_{\min}\) is additionally imposed between tones retained in
the same control, consistent with the multitone separation condition
discussed in Sec.~\ref{sec:linear_control_constraints}.

The nonidentity diagonal interaction space of an \(n\)-spin register has
dimension \(2^n-1\). We therefore take
\begin{equation}
M_{\max}
=
2^n-1
\label{eq:default_maximum_tone_budget}
\end{equation}
as the default maximum tone budget of the finite-tone construction. This
sets the control budget by the dimension of the interaction space being
shaped rather than by a register-size-independent number of tones. A
smaller value may be imposed when additional hardware restrictions are
present.

Starting from an empty set, \(A=\emptyset\), we construct the tone set
iteratively from \(\Omega_{\mathrm{cand}}\). At any stage of the
construction, \(A\subset\Omega_{\mathrm{cand}}\) denotes the currently
retained tone frequencies, and \(G_A\) the transfer matrix formed from
those frequencies.
To realize the target-channel condition introduced above
using non-negative drive intensities, we test both signed unit
representatives of the target direction,
\begin{equation}
\mathbf K^{\mathrm{des}}_{q}
=
q\,\mathbf e_{S^\ast},
\qquad
q\in\{-1,+1\},
\label{eq:signed_target_vectors}
\end{equation}
where \(\mathbf e_{S^\ast}\) has unit coefficient in the target channel
\(S^\ast\) and zero coefficients in all spectator channels. Thus
\[
K_{S^\ast}^{\mathrm{des}}=q,
\qquad
K_{S}^{\mathrm{des}}=0
\quad
(S\neq S^\ast).
\]
The overall magnitude is fixed only provisionally at this stage, because
the resulting intensity vector is subsequently rescaled to satisfy the
physical dispersive bound.

For each sign \(q\), the relative tone intensities are obtained from the
non-negative inversion
\begin{equation}
\mathbf x_{A,q}^{(u)}
=
\arg\min_{\mathbf x\ge0}
\left\|
G_A\mathbf x
-
q\,\mathbf e_{S^\ast}
\right\|_2 .
\label{eq:target_control_inversion}
\end{equation}
Both signs of the target interaction channel are therefore considered
explicitly. For each trial tone set, the solutions for \(q=+1\) and
\(q=-1\) are both evaluated, and the better solution according to the
selection criteria defined below is retained. The non-negativity
constraint follows directly from the intensity definition in
Eq.~\eqref{eq:intensity_variable}.

Since only a limited number of tones is retained, the tone set is built
iteratively. For the current retained set \(A\), define the set of
remaining spacing-compatible candidates as
\begin{equation}
\mathcal C(A)
=
\left\{
\omega_j\in\Omega_{\mathrm{cand}}\setminus A
:
|\omega_j-\omega|
\ge
\delta\omega_{\min}
\quad
\forall\,\omega\in A
\right\}.
\label{eq:remaining_tone_candidates}
\end{equation}
For every \(\omega_j\in\mathcal C(A)\), we form the one-tone trial extension
\begin{equation}
A_j
=
A\cup\{\omega_j\},
\label{eq:trial_tone_set}
\end{equation}
construct \(G_{A_j}\), and solve
Eq.~\eqref{eq:target_control_inversion} for both
\(q\in\{-1,+1\}\). Thus, at each iteration, every remaining
spacing-compatible tone is tested individually as one possible addition
to the current retained set \(A\), with the non-negative intensities of
the entire extended set \(A_j\) re-solved for both target signs. The
resulting one-tone trial extensions are compared using the selection
criteria defined below. Denoting the preferred trial at the current
iteration by \(A_{j^\ast}\), we update
\(A\leftarrow A_{j^\ast}\), reconstruct \(\mathcal C(A)\) for the updated
set, and repeat the procedure. The remaining frequencies are therefore
re-evaluated in combination with the already retained tones at every
iteration, rather than being assigned a fixed score from their individual
transfer-matrix columns.

For each trial set \(A_j\) and target sign \(q\), the corresponding
interaction vector is
\begin{equation}
\mathbf K_{A_j,q}^{(u)}
=
G_{A_j}\mathbf x_{A_j,q}^{(u)} .
\label{eq:trial_interaction_vector}
\end{equation}
The quality of the inversion is quantified by the residual
\begin{equation}
r_{\mathrm{inv}}(A_j,q)
=
\left\|
G_{A_j}\mathbf x_{A_j,q}^{(u)}
-
q\,\mathbf e_{S^\ast}
\right\|_2 ,
\label{eq:inversion_residual}
\end{equation}
which directly measures how closely the trial control reproduces the
specific desired interaction pattern. For each trial tone set, both signs
\(q\in\{-1,+1\}\) are evaluated and the preferred signed solution is
retained. The inversion residual provides the central measure of agreement
with the desired interaction pattern.

If additional discrimination between candidate solutions is useful, the
selection can be refined using further practical criteria. For a single
target channel, spectator suppression can be quantified by
\begin{equation}
\epsilon_{\mathrm{spec}}(A_j,q)
=
\frac{
\left\|
\mathbf K_{A_j,q,\overline{S^\ast}}^{(u)}
\right\|_2
}{
\left|
K_{A_j,q,S^\ast}^{(u)}
\right|
},
\label{eq:spectator_ratio}
\end{equation}
where
\(\mathbf K_{A_j,q,\overline{S^\ast}}^{(u)}\)
contains all interaction coefficients except the target coefficient.
Likewise, the total unscaled intensity
\begin{equation}
X_{\mathrm{tot}}(A_j,q)
=
\mathbf 1^{\mathsf T}\mathbf x_{A_j,q}^{(u)}
=
\sum_{k=1}^{|A_j|} x_{A_j,q,k}^{(u)} .
\label{eq:total_unscaled_intensity}
\end{equation}
Here \(\mathbf 1\) denotes the all-ones vector of dimension \(|A_j|\).
This quantity may be used as an additional practical criterion. The
spectator ratio and total unscaled intensity provide optional refinements
of the candidate selection, while the inversion residual remains the direct
measure of how accurately the desired interaction pattern is reproduced.

The inversion determines the relative distribution of intensity among
the selected tones but does not yet fix the absolute physical drive
strength. Once the final tone set and target sign have been selected, we
denote the corresponding unscaled intensity vector by
\(\mathbf x^{(u)}\) and apply a common positive scale,
\begin{equation}
\mathbf x
=
\zeta\,\mathbf x^{(u)} .
\label{eq:scaled_intensities}
\end{equation}
Using the nearest-transition detunings \(d_k\) defined in
Eq.~\eqref{eq:theory_nearest_detuning} and the dispersive intensity bound
of Eq.~\eqref{eq:dispersive_intensity_bound}, the largest common scale
consistent with the prescribed dispersive ratio \(r_{\mathrm{disp}}\) is
\begin{equation}
\zeta
=
\min_{k:x_k^{(u)}>0}
\frac{
r_{\mathrm{disp}}^2 d_k^2
}{
2x_k^{(u)}
}.
\label{eq:practical_dispersive_scaling}
\end{equation}
This global rescaling leaves the relative intensity distribution obtained
from the control inversion unchanged while placing the control at the
chosen dispersive bound.

Once the carrier frequencies and physical intensities
have been fixed, the linear control map determines the physical
interaction rates \(\mathbf K=G_A\mathbf x\). The coefficient \(K_S\)
generates the conditional Pauli-string phase at rate \(K_S/2\). For a
mediator fixed in a \(\sigma_z\) eigenstate with
eigenvalue \(s_{\mathrm{med}}=\pm1\), the corresponding register phase
rate is
\(\dot{\Phi}_S=s_{\mathrm{med}}K_S/2\), as established in
Corollary~\ref{cor:effective_evolution}.

For the selected target sign \(q\) and a mediator prepared in the
eigenstate \(s_{\mathrm{med}}\), the corresponding register phase has
sign \(s_{\mathrm{med}}q\). We therefore take
\begin{equation}
\Phi_{S^\ast}^{\mathrm{tar}}
=
s_{\mathrm{med}}q\,
\left|\Phi_{S^\ast}^{\mathrm{tar}}\right|.
\end{equation}
The required positive interaction time is then
\begin{equation}
T
=
\frac{
2\left|\Phi_{S^\ast}^{\mathrm{tar}}\right|
}{
\left|K_{S^\ast}\right|
}.
\label{eq:gate_time_from_optimized_rate}
\end{equation}
Once \(T\) has been fixed, the omitted identity coefficient
\(K_{\emptyset}\) is evaluated from
Eq.~\eqref{eq:identity_channel_coefficient}, yielding the known virtual
mediator-frame correction
\(\phi_{\mathrm{med}}=K_{\emptyset}T\).

Thus the inversion fixes the relative interaction pattern, the dispersive
rescaling fixes the physical drive strength, the evolution time fixes the
accumulated target phase, and the separately evaluated identity coefficient
determines the virtual mediator-frame correction.

The dispersive constraint defines the perturbative operating regime but
does not by itself guarantee the accuracy of the full driven mediator-register dynamics. We therefore
test the synthesized controls independently against the full dynamics in
Sec.~\ref{sec:numerical-demonstrations}.

Algorithm~\ref{alg:practical_finite_tone_synthesis} summarizes the
finite-tone construction.

\begin{algorithm}[Target-aware finite-tone control synthesis]
\label{alg:practical_finite_tone_synthesis}
\begin{algorithmic}[1]

\Statex
\parbox[t]{0.95\linewidth}{
\textbf{Input:}
Transition spectrum \(\{\Lambda_\alpha\}\);
target interaction \(Z_{S^\ast}\);
target-phase magnitude
\(\left|\Phi_{S^\ast}^{\mathrm{tar}}\right|\);
mediator eigenvalue \(s_{\mathrm{med}}\in\{-1,+1\}\);
dimensionless tone-spacing fraction \(\eta_\omega\) (default \(=0.075\));
dispersive ratio \(r_{\mathrm{disp}}\) (default \(=0.10\)).
}.

\Statex
\parbox[t]{0.95\linewidth}{
\textbf{Output:}
Selected frequencies \(\boldsymbol{\omega}_A\);
selected target sign \(q^\ast\);
physical intensities \(\mathbf x\);
interaction rates \(\mathbf K\);
interaction time \(T\);
virtual mediator-frame correction angle \(\phi_{\mathrm{med}}\).
}

\State Compute the characteristic mean transition spacing
\(\bar{\Delta}\) and minimum adjacent transition spacing
\(\Delta_{\min}\) from
Eq.~\eqref{eq:characteristic_transition_spacings}.

\State Set
\(d_{\min}=0.45\,\Delta_{\min}\) and
\(\delta\omega_{\min}=\eta_\omega\bar{\Delta}\)
according to Eq.~\eqref{eq:normalized_frequency_resolution}.

\State Set the default tone budget
\(M_{\max}=2^n-1\) according to
Eq.~\eqref{eq:default_maximum_tone_budget}.

\State Taking the spectral center as the frequency origin, compute
\(B=\max_\alpha|\Lambda_\alpha|\) and set
\(\mathcal W=[-1.35B,\,1.35B]\) according to
Eq.~\eqref{eq:default_frequency_window}.

\State Uniformly discretize \(\mathcal W\) into
\(N_{\mathrm{grid}}=6401\) frequencies to form
\(\Omega_{\mathrm{grid}}\).

\State Construct \(\Omega_{\mathrm{cand}}\) using
Eq.~\eqref{eq:candidate_tone_pool}.

\State Initialize the selected tone set
\(A\leftarrow\emptyset\).

\While{\(|A|<M_{\max}\)}

    \State Construct the remaining spacing-compatible candidate set
    \(\mathcal C(A)\) using
    Eq.~\eqref{eq:remaining_tone_candidates}.

    \If{\(\mathcal C(A)=\emptyset\)}
        \State \textbf{break}
    \EndIf

    \For{each \(\omega_j\in\mathcal C(A)\)}

        \State Form
        \(A_j=A\cup\{\omega_j\}\)
        and construct the corresponding transfer matrix \(G_{A_j}\).

        \State Solve
        Eq.~\eqref{eq:target_control_inversion}
        for both \(q\in\{-1,+1\}\).

        \State For each sign compute
        \(r_{\mathrm{inv}}(A_j,q)\) from
        Eq.~\eqref{eq:inversion_residual}.

        \State Retain the preferred sign
        \(q_j^\ast=\arg\min_{q\in\{-1,+1\}}
        r_{\mathrm{inv}}(A_j,q)\).

    \EndFor

    \State For each trial set define
    \(r_j=r_{\mathrm{inv}}(A_j,q_j^\ast)\).

    \State Rank the trial sets primarily by increasing \(r_j\).
    If additional discrimination is desired, refine this ranking using
    \(\epsilon_{\mathrm{spec}}\) from
    Eq.~\eqref{eq:spectator_ratio} and/or
    \(X_{\mathrm{tot}}\) from
    Eq.~\eqref{eq:total_unscaled_intensity}.

    \State Let \(j^\ast\) denote the preferred candidate under this
    ranking and update
    \(A\leftarrow A\cup\{\omega_{j^\ast}\}\).

    \State Retain the corresponding signed solution and relative
    intensity vector \(\mathbf x^{(u)}\).

\EndWhile

\State Let \(q^\ast\) and \(\mathbf x^{(u)}\) be those of the final
accepted trial set and remove tones assigned zero intensity in
\(\mathbf x^{(u)}\).

\State Apply the common dispersive scale from
Eq.~\eqref{eq:practical_dispersive_scaling},
\(\mathbf x=\zeta\mathbf x^{(u)}\).

\State Compute the physical interaction rates
\(\mathbf K=G_A\mathbf x\) using
Eq.~\eqref{eq:linear_hamiltonian_engineering_system}.

\State Set
\(\Phi_{S^\ast}^{\mathrm{tar}}
=s_{\mathrm{med}}q^\ast
\left|\Phi_{S^\ast}^{\mathrm{tar}}\right|\)
and determine \(T\) from
Eq.~\eqref{eq:gate_time_from_optimized_rate}.

\State Determine \(K_{\emptyset}\) from
Eq.~\eqref{eq:identity_channel_coefficient} and set the virtual
mediator-frame correction angle
\(\phi_{\mathrm{med}}=K_{\emptyset}T\)
according to Eq.~\eqref{eq:mediator_phase}.

\State \Return
\(\boldsymbol{\omega}_A,q^\ast,\mathbf x,\mathbf K,T,\phi_{\mathrm{med}}\).

\end{algorithmic}
\end{algorithm}

\section{Numerical benchmarks}
\label{sec:numerical-demonstrations}

We now quantify the performance of the control construction developed from
the linear dispersive framework above. We refer to this approach as
\emph{dispersive interaction via analytical linear inversion} (DIAL).
For this purpose, we use the target-aware DIAL synthesis algorithm developed
in Sec.~\ref{sec:phase_rate_optimization}. The benchmarks determine how accurately the direct DIAL construction
realizes prescribed multi-qubit interactions across different register spectra
and register sizes while suppressing mediator excitation. In particular, we
test controls obtained entirely from the analytical construction under the
full driven mediator-register dynamics, without subsequent numerical control
optimization.

The DIAL controls are constructed entirely from the dispersively eliminated
linear model. To assess their performance beyond this approximation, we
propagate them under the full dynamics generated by
Eq.~\eqref{eq:general_rwa_hamiltonian}, without performing the dispersive
elimination and with the mediator retained explicitly. Hereafter we refer to
this propagation simply as the \emph{full dynamics}. The resulting fidelity
therefore assesses the accuracy of a control designed from the effective
linear model when applied to the underlying driven mediator-register system,
while the transient mediator population directly quantifies the extent to
which mediator excitation remains suppressed during the intended dispersive
operation.

For the direct-synthesis benchmark, we consider registers with
\(n=2,3,4,5\) spins and sample \(N=100\) random spectral realizations for
each register size from the normalized benchmark ensemble defined in
Appendix~\ref{app:benchmark_ensemble}. For every realization, we target the
\(n\)-body Pauli-\(Z\) interaction \(Z_1Z_2\cdots Z_n\) with interaction
phase
\(\left|\Phi_{Z_1\cdots Z_n}^{\rm tar}\right|=\pi/4\), for which the
corresponding Pauli-string evolution transforms
\(|+\rangle^{\otimes n}\) into a state locally equivalent to a GHZ state.

For each sampled spectrum, we construct the corresponding DIAL control using
the strongest drive amplitudes compatible with the prescribed dispersive
regime, as determined by the dispersiveness tolerance \(r_{\rm disp}\)
defined in Eq.~\eqref{eq:practical_dispersive_scaling} and supplied as an
input to Algorithm~\ref{alg:practical_finite_tone_synthesis}. The resulting
interaction rate then fixes the gate duration required to accumulate the
prescribed target phase. We propagate this control under the full dynamics
defined above and record the resulting fidelity and maximum transient
mediator excitation during the evolution.

For all register sizes, we use the same dimensionless DIAL parameters,
\(\eta_d=\eta_\omega=0.075\) and \(r_{\rm disp}=0.10\), together with the
default tone budget prescribed by
Algorithm~\ref{alg:practical_finite_tone_synthesis}. Thus, the benchmark uses
a common synthesis prescription across register sizes without
register-size-specific tuning of the frequency constraints or tone budget.

\subsection{Direct DIAL synthesis across registers and sizes}
\label{sec:direct_inversion_scaling}
\label{subsec:full-rwa-validation}

We first quantify the interaction generated directly by DIAL.
Figure~\ref{fig:direct_inversion_scaling} summarizes the resulting
performance across the sampled register spectra and different register
sizes. Each point corresponds to a separately sampled register spectrum
and to the control constructed by the Target-Aware DIAL algorithm.

\begin{figure*}[t]
    \centering
    \includegraphics[width=\textwidth]
    {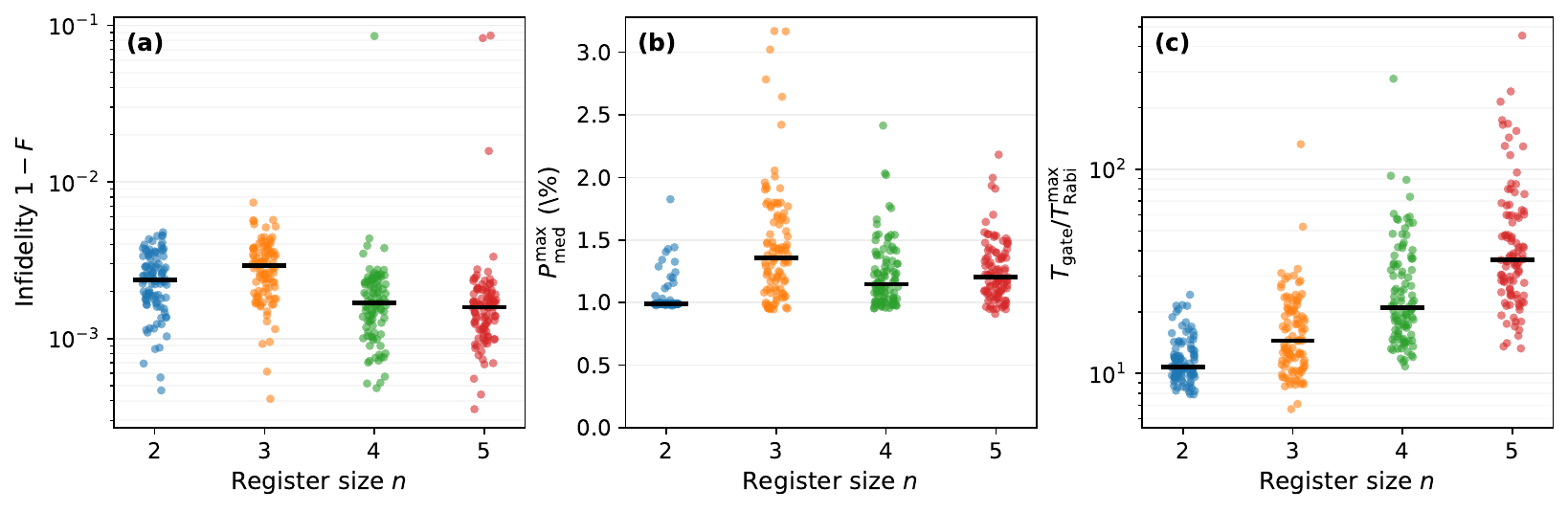}
    \caption{
    Direct DIAL synthesis across random register spectra and different
    register sizes.
    (a) Full-dynamics infidelity \(1-F\) of the DIAL controls for
    \(n=2,3,4,5\), shown on a logarithmic scale. Each point represents
    one sampled register spectrum and targets the highest-order
    interaction \(Z_1\cdots Z_n\) at
    \(\lvert\Phi_{Z_1\cdots Z_n}^{\mathrm{tar}}\rvert=\pi/4\).
    The higher-infidelity outlier can be explained by an exceptionally
    crowded register spectrum.
    (b) Maximum transient mediator excitation
    \(P_{\mathrm{med}}^{\max}\) for the same controls, with the
    individual realizations retained.
    (c) Gate duration normalized by the Rabi period of the strongest
    applied tone,
    \(T_{\mathrm{Rabi}}^{\max}=2\pi/\Omega_{\max}\), with
    \(\Omega_{\max}=\max_k|\Omega_k|\).
    Horizontal bars indicate medians and individual points denote
    sampled register spectra. All controls are constructed directly
    from the Target-Aware DIAL algorithm and evaluated under the full
    dynamics without subsequent numerical control optimization.
    }
    \label{fig:direct_inversion_scaling}
\end{figure*}

For a target \(n\)-body interaction, DIAL determines the drive tones and
intensities, the resulting interaction rates, and the interaction time
required to accumulate the prescribed target phase. In the effective
description, the corresponding target evolution has the form
\begin{equation}
\begin{aligned}
U_{\mathrm{tar}}(T_{\mathrm{gate}})
&=
e^{-i\phi_{\mathrm{med}}
\sigma_z/2}
\\
&\quad\times
\exp\!\left[
-i q^\ast \frac{\pi}{4}
\sigma_z
\otimes
Z_1\cdots Z_n
\right],
\end{aligned}
\end{equation}
where
\(\phi_{\mathrm{med}}=K_{\emptyset}T_{\mathrm{gate}}\)
is the known local mediator phase and
\(q^\ast\in\{-1,+1\}\) is the interaction sign selected by the
Target-Aware DIAL algorithm.

Since the relevant operation for the present benchmark is the induced
interaction on the register, we initialize the mediator in its ground
state \(\lvert g\rangle\), with
\(\sigma_z\lvert g\rangle=-\lvert g\rangle\).
Within this mediator manifold, the local mediator rotation contributes
only an overall phase, while the target register evolution reduces to
\begin{equation}
    U_{\mathrm{tar}}^{(g)}
    =
    \exp\!\left[
        +i q^\ast \frac{\pi}{4}
        Z_1\cdots Z_n
    \right].
\label{eq:ground_manifold_target}
\end{equation}
Choosing the opposite mediator eigenstate produces the same interaction
magnitude with the opposite known sign.

We quantify the accuracy of the full driven dynamics through the
ground-manifold register fidelity
\begin{equation}
    F
    =
    \frac{1}{4^n}
    \left|
        \operatorname{Tr}_{\mathrm{reg}}
        \left[
            \left(U_{\mathrm{tar}}^{(g)}\right)^\dagger
            \langle g|
            U(T_{\mathrm{gate}})
            |g\rangle
        \right]
    \right|^2 .
\label{eq:ground_manifold_fidelity}
\end{equation}
Because the projected register propagator
\(\langle g|U(T_{\mathrm{gate}})|g\rangle\)
is obtained from the full mediator-register evolution, population that
fails to return to the ground mediator manifold is automatically
penalized by \(F\). Transient mediator excitation during the evolution is
characterized independently through \(P_{\mathrm{med}}^{\max}\).

As shown in Fig.~\ref{fig:direct_inversion_scaling}(a), the typical
full-dynamics performance remains high across all four register sizes.
The median fidelities are approximately \(0.9976\), \(0.9971\),
\(0.9983\), and \(0.9984\) for \(n=2,3,4,5\), respectively. Thus there
is no systematic loss of typical fidelity across the different register
sizes considered here.

To quantify mediator excitation during the gate, we evaluate the
maximum excited-state population over the full evolution and over all
computational register configurations,
\begin{equation}
\begin{aligned}
P_{\mathrm{med}}^{\max}
&=
\max_{\alpha}
\max_{0\le t\le T_{\mathrm{gate}}}
\bra{g,\alpha}
U^\dagger(t)
\\
&\quad\times
\left(
\ket{e}\!\bra{e}\otimes \mathds{1}_{\mathrm{reg}}
\right)
U(t)
\ket{g,\alpha},
\end{aligned}
\label{eq:max_transient_mediator_excitation}
\end{equation}
where \(\alpha\) labels the computational register configurations.
Figure~\ref{fig:direct_inversion_scaling}(b) shows the corresponding
maximum transient mediator excitation under the full dynamics. The median
maximum mediator excitation is approximately \(0.99\%\), \(1.36\%\),
\(1.15\%\), and \(1.20\%\) for \(n=2,3,4,5\), respectively. The
typical transient mediator population therefore remains at the percent
level and shows no monotonic increase across the different register
sizes considered here.

The percent-level mediator population observed across the benchmark is
consistent with the chosen dispersive operating regime,
\(r_{\mathrm{disp}}=0.1\), corresponding to an isolated-transition
population scale of \(\sim 1\%\), as quantified in
Proposition~\ref{prop:dispersivity_mediator_excitation}.

At this level of dispersivity, the corresponding interaction times are
also naturally expected to be long compared with the applied Rabi
oscillation periods.

To compare the interaction times of the synthesized controls on a
dimensionless scale, we express the gate duration in units of the Rabi
period associated with the strongest applied tone. For each control, we
define
\begin{equation}
T_{\mathrm{Rabi}}^{\max}
=
\frac{2\pi}{\Omega_{\max}},
\qquad
\Omega_{\max}
=
\max_k |\Omega_k|,
\label{eq:strongest_tone_rabi_period}
\end{equation}
so that \(T_{\mathrm{gate}}/T_{\mathrm{Rabi}}^{\max}\) measures the
interaction time relative to the strongest applied drive.
Figure~\ref{fig:direct_inversion_scaling}(c) reports the corresponding
gate duration in units of the strongest-tone Rabi period. The median values
of \(T_{\mathrm{gate}}/T_{\mathrm{Rabi}}^{\max}\) are approximately
\(10.8\), \(14.5\), \(21.0\), and \(36.0\) for \(n=2,3,4,5\),
respectively. In contrast to the fidelity and mediator-excitation
distributions, the required normalized interaction time increases
substantially across the different register sizes considered here.

Together, the three panels show that DIAL directly produces
excitation-suppressed, high-fidelity multi-qubit controls across the sampled
register spectra. The clearest systematic change occurs in the interaction
time: the normalized gate-duration distribution shifts toward longer
controls and develops a broader upper range across the different register
sizes. The relation between spectral crowding and the required normalized
gate duration is examined directly in
Appendix~\ref{app:scaling_spectral_geometry}.

\section{Summary and outlook}
\label{sec:summary}

We have shown that, in the configuration-preserving dispersive regime of
mediator-coupled quantum registers, the interactions induced by
off-resonant multi-tone driving admit an explicit linear control
description. Starting from the driven mediator-register Hamiltonian, we
derive the dispersive effective dynamics, express the resulting
register-dependent Stark shifts as a diagonal Hamiltonian, and decompose
this Hamiltonian into Pauli-\(Z\) interaction channels. For fixed tone
frequencies and constant-amplitude driving, the corresponding
nontrivial interaction rates obey
\begin{equation}
\mathbf K = G\mathbf x ,
\end{equation}
where the non-negative vector \(\mathbf x\) contains the tone intensities
and the transfer matrix \(G\) is determined by the register-resolved
mediator spectrum and the chosen tone frequencies. The identity
coefficient \(K_{\emptyset}\) follows from the same control but generates
only a known local mediator \(Z\) rotation. Together with the physical
constraints imposed by dispersive driving and multi-tone control, this
establishes a direct linear mapping from experimentally accessible
controls to the induced diagonal multi-qubit interactions.

We refer to this control principle as
\emph{dispersive interaction via analytical linear inversion} (DIAL).
Its central consequence is that the dispersive effective Hamiltonian can
be used not only as a forward model for predicting the interactions
generated by a given drive, but also as a control-design model. Once the
register-resolved mediator spectrum has been characterized, a desired
Pauli-\(Z\) interaction pattern can be translated into a constrained
linear problem for the applied tone intensities. This provides an
analytical route from spectral information to engineered multi-qubit
interactions without requiring an iterative search over the full driven
dynamics.

We turn this principle into the Target-Aware DIAL algorithm, a practical
finite-tone synthesis procedure for a prescribed Pauli-\(Z\)
interaction. The algorithm selects physically admissible off-resonant
tones and determines their drive intensities through a constrained
linear inversion that enhances the target channel while suppressing
spectator interactions. A common rescaling places the resulting control
within the chosen dispersive regime, after which the interaction time is
set from the synthesized target rate to accumulate the prescribed phase.
The complete control prescription, including tone frequencies, intensities, interaction rates, and interaction time, is therefore obtained directly
from the characterized transition spectrum.

We tested this direct construction on random register-resolved spectra
for \(n=2,3,4,5\), drawn from the normalized benchmark ensemble defined
in Appendix~\ref{app:benchmark_ensemble}, targeting the highest-order
interaction \(Z_1\cdots Z_n\) with
\(\lvert\Phi_{Z_1\cdots Z_n}^{\mathrm{tar}}\rvert=\pi/4\).
Propagation of the resulting controls under the full driven
mediator-register dynamics gives median fidelities of approximately
\(0.9976\), \(0.9971\), \(0.9983\), and \(0.9984\), respectively,
while the corresponding median maximum transient mediator excitations
are \(0.99\%\), \(1.36\%\), \(1.15\%\), and \(1.20\%\). Thus, across
the register sizes and spectral realizations considered here, the
analytically constructed controls retain high typical full-dynamics
fidelity while keeping mediator excitation at the percent level.

The clearest systematic change in the benchmark is the required
interaction time. The median values of
\(T_{\mathrm{gate}}/T_{\mathrm{Rabi}}^{\max}\)
increase from approximately \(10.8\) for \(n=2\) to \(36.0\) for
\(n=5\), with the longest-duration realizations associated with locally
crowded register-resolved spectra, as analyzed in
Appendix~\ref{app:scaling_spectral_geometry}. This connects the performance of DIAL
directly to the spectral geometry of the mediator-register system:
when nearby transitions become harder to distinguish dispersively,
selective interaction engineering can require longer interaction times.
The linear control formulation therefore exposes not only how a desired
interaction can be constructed, but also the spectral constraints that
limit that construction.

The percent-level mediator excitation together with the comparatively
long interaction times observed in the benchmark are consistent with the
dispersive population-gate-time tradeoff given in
Proposition~\ref{prop:dispersivity_gate_time_tradeoff} and
Corollary~\ref{cor:mediator_population_gate_time_tradeoff}.

The present formulation focuses on the configuration-preserving
diagonal dispersive sector, a common regime that arises naturally in many
mediator-coupled systems or can be approached in a high-field or secular
limit. A natural extension is therefore to determine
how the control map itself can incorporate inter-manifold basis mismatch
and the associated register-changing channels. The effect of such mismatch
on controls designed under the nominal configuration-preserving model is
examined in Appendix~\ref{app:misalignment_audit}. Although DIAL does not
explicitly account for this mismatch, the resulting controls remain
accurate for small to moderate deviations while mediator excitation stays
suppressed. This suggests that DIAL may also
provide a useful analytical initialization for subsequent numerical
optimal control when the mediator-register model is only approximately
calibrated. Further directions include decoherence, spectral uncertainty,
finite control bandwidth, additional hardware constraints, shaped drive
envelopes, and improved tone-selection strategies for larger registers.

\bibliographystyle{unsrt}

\bibliography{references}

\section*{Acknowledgements}

This work was supported by the European Union's HORIZON Europe program via projects SPINUS (No.~101135699), QCFD (No.~101080085, HORIZON-CL4-2021-DIGITAL-EMERGING-02-10), OpenSuperQ-Plus100 (No.~101113946, HORIZON-CL4-2022-QUANTUM-01-SGA) and PASQuanS2.1 (101113690,HORIZON-CL4-2022-QUANTUM-02-SGA), by AIDAS-AI, Data Analytics and Scalable Simulation, which is a Joint Virtual Laboratory gathering the Forschungszentrum Juelich and the French Alternative Energies and Atomic Energy Commission, and by Germany's Excellence Strategy -- the Cluster of Excellence Matter and Light for Quantum Computing (ML4Q2) EXC 2004/2--390534769.
J.O. acknowledges support from the France 2030 plan under the ROBUSTSUPERQ (ANR-22-PETQ-0003) grant and from the European Research Council under grant no.~101201502 (nQUICHE).
The authors declare no competing financial interest.

\section*{Author contributions}
\noindent
The research question was conceived and refined in discussions between B.B. and all other authors. B.B. led the theoretical development, developed the method and implemented the method. T.C. and J.O. provided critical feedback and scientific guidance, while F.M. and M.M.M. supervised the project. B.B. drafted the manuscript, and all authors contributed to its final form.

\section*{Data availability}
\noindent
The data and reproduction code supporting this work are available in the
accompanying DIAL repository~\cite{baran2026dial}.

\appendix

\section{Benchmark ensemble in normalized spectral space}
\label{app:benchmark_ensemble}

The specific system enters only via the register-resolved mediator
transition frequencies \(\Lambda_\alpha\) introduced in the
configuration-preserving reduction above. To test the control construction
across a diverse range of possible register systems, we sample random
realizations via the register-dependent transition frequencies of the
mediator two-level system. For an \(n\)-spin register, we parameterize
these transition frequencies as
\begin{equation}
\Lambda_\alpha
=
\Lambda_0+\delta_\alpha ,
\label{eq:appendix_transition_parameterization}
\end{equation}
where \(\Lambda_0\) is the reference transition frequency of the
mediator and \(\delta_\alpha\) denotes the
register-configuration-dependent shift of this transition. Here
\(\alpha=(\alpha_1,\ldots,\alpha_n)\), with
\(\alpha_i\in\{0,1\}\), labels the register configuration. The
\(2^n\) possible register configurations therefore generate
\(2^n\) register-resolved transition frequencies of the mediator.

For the benchmark ensemble, we choose
\begin{equation}
\delta_\alpha
=
\Lambda_{\mathrm{sc}}
\frac{
\displaystyle\sum_{i=1}^{n}
(-1)^{\alpha_i}A_i
}{
\displaystyle\sum_{i=1}^{n}A_i
}.
\label{eq:appendix_normalized_spectrum}
\end{equation}
Here \(A_i>0\) specifies the relative contribution of register spin
\(i\) to the mediator-transition splitting, while
\(\Lambda_{\mathrm{sc}}\) fixes the overall spectral scale. The normalization
by \(\sum_i A_i\) separates this overall scale from the relative
spectral structure.

Because only the relative values of the \(A_i\) determine this
normalized spectral structure, we fix the reference scale to
\(A_1=1\). We additionally order the contributions according to
\begin{equation}
1=A_1\leq A_2\leq\cdots\leq A_n\leq 10.
\label{eq:appendix_ordered_A}
\end{equation}
The normalized spectrum of an \(n\)-spin realization is therefore
specified by the remaining \(n-1\) coordinates
\((A_2,\ldots,A_n)\). Regions excluded by this ordering correspond to
permutations of the register-spin labels and are therefore not sampled
separately.

The spectrum is determined by the relative sizes of the \(A_i\), rather
than by their absolute differences. We therefore sample the independent
relative contributions logarithmically to obtain a broad variation of
distinct spectral structures without preferentially oversampling one
end of the allowed interval. For example, changing a contribution from
\(1\) to \(2\) represents the same factor-of-two change in relative
spectral weight as changing it from \(5\) to \(10\).

For \(n>3\), the resulting benchmark ensemble occupies an
\((n-1)\)-dimensional normalized spectral space and therefore does not
admit the same direct two-dimensional visualization. The \(n=3\) case,
however, is completely specified by \(A_2\) and \(A_3\) and provides a
simple visual example of the general sampling construction.
Figure~\ref{fig:benchmark_spectral_ensemble_n3} shows the actual
\(n=3\) realizations used in the direct DIAL benchmark.

\begin{figure}[t]
    \centering
    \includegraphics[width=\columnwidth]
    {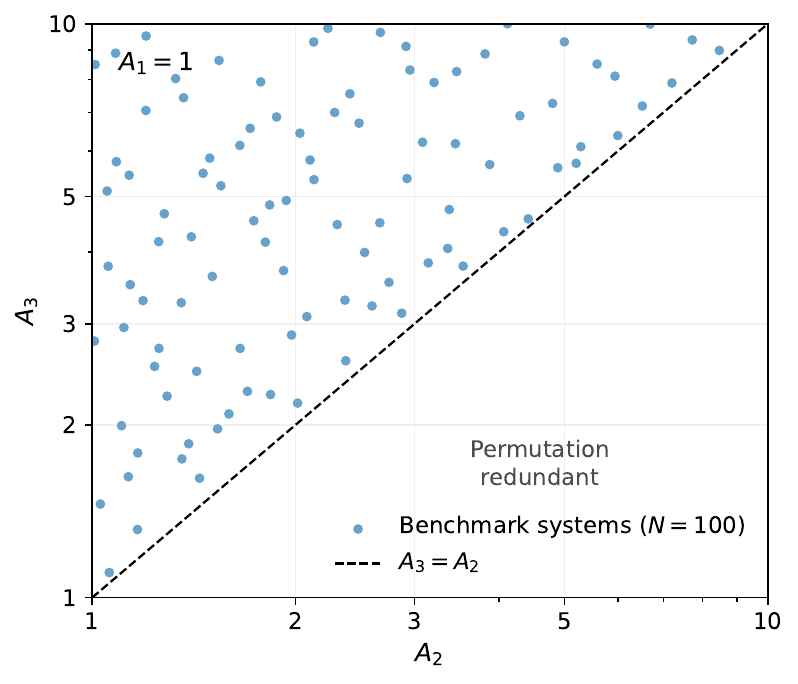}
    \caption{
    Distribution of the \(n=3\) transition spectra used in the direct DIAL
    benchmark, shown as a two-dimensional example of the normalized
    \(n\)-spin benchmark ensemble. The relative spectral contributions
    are normalized by \(A_1=1\) and ordered as
    \(1=A_1\leq A_2\leq A_3\leq10\). Both axes are logarithmic,
    matching the sampling of the relative spectral contributions. The
    region \(A_3<A_2\) is permutation redundant: the ratios
    \(A_2/A_1\) and \(A_3/A_1\) specify the relative spectral contributions
    that determine the transition spectrum, so exchanging register-spin
    labels \(2\leftrightarrow3\) only swaps these ratios and does not
    produce a physically distinct system.
    }
    \label{fig:benchmark_spectral_ensemble_n3}
\end{figure}

\section{Spectral crowding and interaction time in the direct DIAL benchmark}
\label{app:scaling_spectral_geometry}

Having defined the normalized benchmark ensemble in Appendix~A, we next
examine how the spectral geometry of the register-resolved transitions
affects the interaction time required by the direct DIAL construction.
The ability of an off-resonant tone to distinguish different register
configurations depends on the separation of their corresponding transition
frequencies. Spectra containing unusually close transition lines therefore
provide a natural setting in which selective dispersive synthesis can
require smaller effective interaction rates and correspondingly longer
gate durations.

A natural quantity for characterizing this spectral crowding is the
separation of the closest pair of transition lines. Let
\(\Lambda_{(1)}\le\Lambda_{(2)}\le\cdots\le\Lambda_{(2^n)}\)
denote the ordered register-resolved transition frequencies of an
\(n\)-spin realization. We define the minimum adjacent line spacing as
\begin{equation}
\Delta_{\min}
=
\min_{1\le m\le 2^n-1}
\left[
\Lambda_{(m+1)}-\Lambda_{(m)}
\right].
\label{eq:minimum_transition_spacing}
\end{equation}
We compare this with the characteristic mean spacing of the same spectrum,
\begin{equation}
\bar{\Delta}
=
\frac{\Lambda_{(2^n)}-\Lambda_{(1)}}{2^n-1}.
\label{eq:mean_transition_spacing}
\end{equation}
The dimensionless ratio
\begin{equation}
\frac{\Delta_{\min}}{\bar{\Delta}}
\label{eq:relative_minimum_spacing}
\end{equation}
is the relative minimum spacing. It provides a bandwidth-independent
measure of local spectral crowding: values substantially below unity
identify spectra containing a pair of transition lines that is unusually
close compared with the characteristic spacing of the full
register-resolved spectrum.

To quantify the interaction time required by the synthesized control, we
use the normalized gate duration introduced in
Eq.~\eqref{eq:strongest_tone_rabi_period},
\begin{equation}
\frac{T_{\mathrm{gate}}}{T_{\mathrm{Rabi}}^{\max}},
\qquad
T_{\mathrm{Rabi}}^{\max}
=
\frac{2\pi}{\Omega_{\max}}.
\end{equation}
This quantity measures the duration of the engineered interaction relative
to the Rabi period associated with the strongest applied tone.

\begin{figure}[t]
    \centering
    \includegraphics[width=\columnwidth]
    {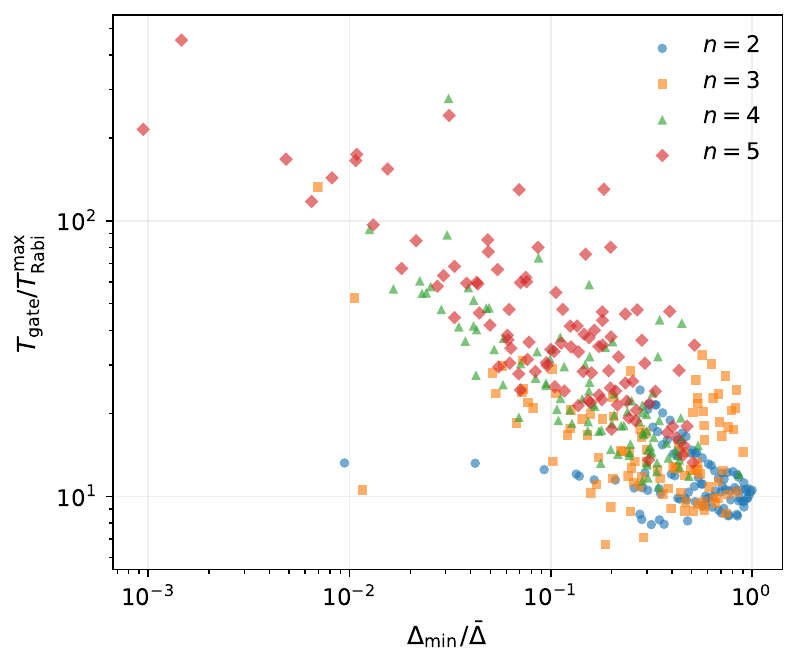}
    \caption{
    Normalized gate duration
    \(T_{\mathrm{gate}}/T_{\mathrm{Rabi}}^{\max}\)
    versus the relative minimum transition spacing
    \(\Delta_{\min}/\bar{\Delta}\) for the random \(n=2,3,4,5\)
    realizations used in the direct DIAL benchmark.
    Both axes are logarithmic.
    Small relative minimum spacing identifies spectra containing an unusually
    close pair of register-resolved transition lines. Spectrally crowded
    realizations extend toward longer normalized gate durations, connecting
    the spectral resolution of the register transitions to the interaction
    rate achieved by the selective dispersive construction.
    }
    \label{fig:scaling_relative_minimum_spacing}
\end{figure}

Figure~\ref{fig:scaling_relative_minimum_spacing} shows that realizations
with small relative minimum spacing extend toward substantially longer
normalized interaction times, whereas well-resolved spectra are
predominantly concentrated at shorter durations. Spectral crowding thus
appears primarily as an interaction-time cost in the sampled benchmark.

This behavior follows directly from the dispersive construction. When two
register-resolved transition frequencies approach one another, their
dispersive denominators become increasingly similar, reducing the spectral
selectivity available for distinguishing the corresponding register
configurations. At the same time, the admissible dispersive drive scale is
constrained by the available detuning from the register-resolved
transitions. Resolving closely spaced spectral structure therefore reduces
the interaction rate available for selective synthesis and increases the
required gate duration. The resulting behavior gives a physically
interpretable spectral-resolution-interaction-rate tradeoff for DIAL.

\section{DIAL performance under inter-manifold basis mismatch}
\label{app:misalignment_audit}

The DIAL control map of
Theorem~\ref{thm:linear_hamiltonian_engineering} and
Eq.~\eqref{eq:linear_hamiltonian_engineering_system} is derived in the
configuration-preserving dispersive regime of
Proposition~\ref{prop:general_dispersive}. Here we assess how controls
constructed under this nominal model perform when the register eigenbases
associated with the two mediator states are allowed to differ. The control
is synthesized at the configuration-preserving point and subsequently held
fixed while a controlled inter-manifold basis mismatch is introduced in the
full dynamics.

We introduce the mismatch through a controlled rotation of the
excited-mediator register Hamiltonian. Taking the ground-manifold register
basis as the reference, we define
\begin{equation}
R(\eta)
=
\bigotimes_{j=1}^{n}
\exp\!\left(-i\frac{\eta}{2}Y_j\right),
\label{eq:misalignment_rotation}
\end{equation}
and replace the nominal excited-manifold Hamiltonian by
\begin{equation}
H_e(\eta)
=
R(\eta)H_e(0)R^\dagger(\eta).
\label{eq:misalignment_excited_hamiltonian}
\end{equation}
The parameter \(\eta\) therefore has a direct geometric meaning: it is
the local rotation angle of each register-spin quantization basis in the
excited mediator manifold relative to the corresponding ground-manifold
basis. The same angle is applied to every register spin so that departure
from configuration preservation is controlled by a single parameter.
Because Eq.~\eqref{eq:misalignment_excited_hamiltonian} is a unitary
transformation, the eigenvalues of \(H_e\) are unchanged. The transformation
therefore changes the relative register eigenbases without simultaneously
changing the excited-manifold spectrum.

Writing the rotated excited-manifold eigenstates as
\(\ket{\beta_e(\eta)}\), the overlap matrix appearing in the full
rotating-wave Hamiltonian,
Eq.~\eqref{eq:general_rwa_hamiltonian}, becomes
\begin{equation}
M_{\beta\alpha}(\eta)
=
\braket{\beta_e(\eta)}{\alpha_g}.
\label{eq:misalignment_overlap}
\end{equation}
At \(\eta=0\), \(R(0)=\mathds{1}\) and
\(M_{\beta\alpha}(0)=\delta_{\beta\alpha}\), recovering the
configuration-preserving dynamics used to derive the dispersive control
map. For \(\eta\neq0\), off-diagonal overlaps appear in
Eq.~\eqref{eq:general_rwa_hamiltonian}, allowing a mediator transition to
be accompanied by a change of register configuration.

We reuse the 100 \(n=3\) DIAL controls from the direct-synthesis
benchmark, which target the \(Z_1Z_2Z_3\) interaction at
\(\lvert\Phi_{Z_1Z_2Z_3}^{\mathrm{tar}}\rvert=\pi/4\) and were
constructed at the nominal configuration-preserving point \(\eta=0\)
using Eq.~\eqref{eq:linear_hamiltonian_engineering_system}. For each
system, the previously determined tone frequencies, intensities, and
interaction time are held fixed throughout the mismatch sweep. For
\(\eta=0^\circ,5^\circ,10^\circ,15^\circ,\) and \(17.5^\circ\),
we propagate these unchanged controls using the full dynamics of
Eq.~\eqref{eq:general_rwa_hamiltonian} evaluated at the corresponding value
of \(\eta\). No recalibration of the transfer matrix and no subsequent
control optimization are performed. This directly tests how the original
DIAL controls perform under inter-manifold basis mismatch.

\begin{figure*}[t]
    \centering
    \includegraphics[width=\textwidth]
    {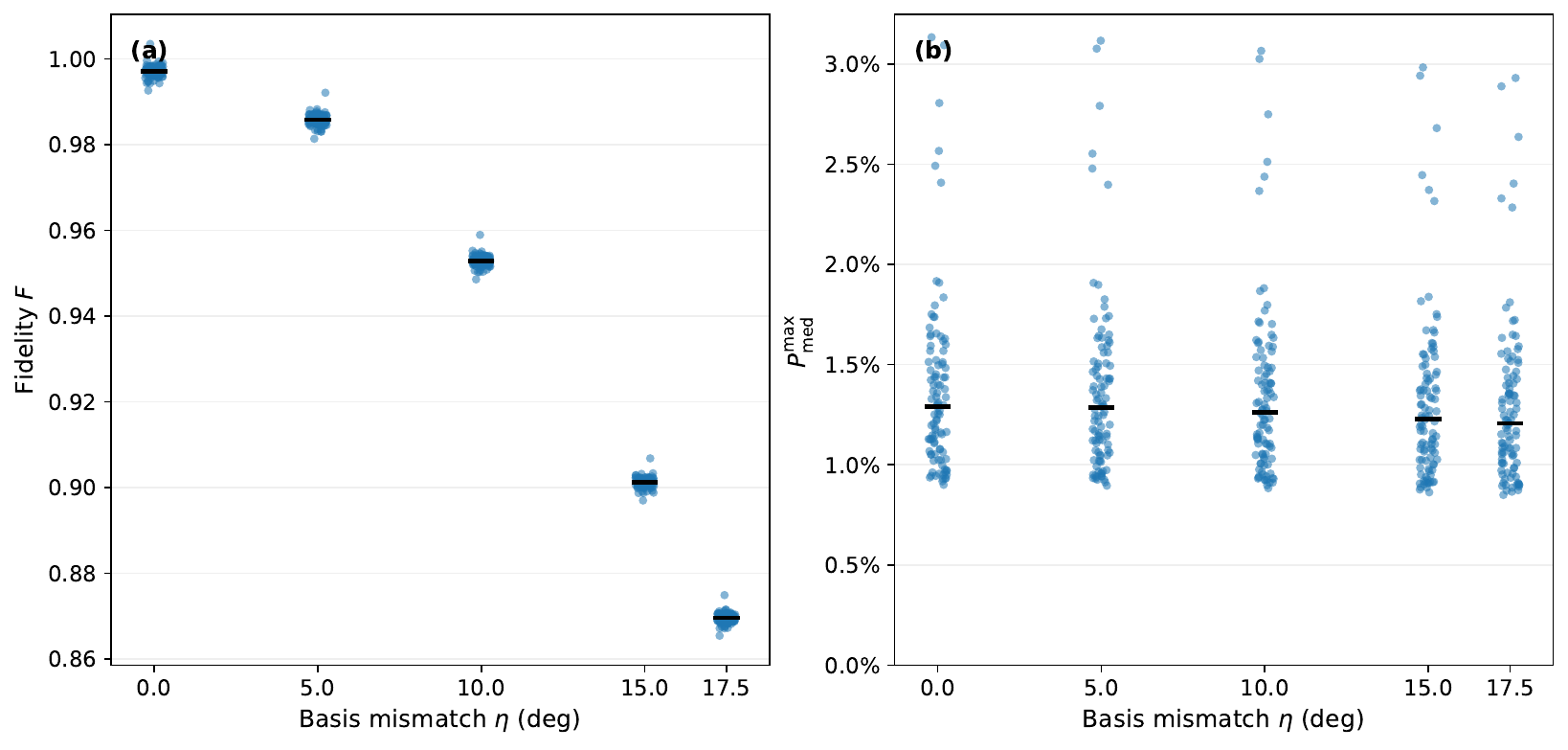}
    \caption{
    Performance of the \(n=3\) direct DIAL controls from the main
    synthesis benchmark under inter-manifold basis mismatch. The same
    100 controls constructed at the nominal configuration-preserving
    point \(\eta=0\) are reused without modification while the
    excited-manifold register basis is rotated according to
    Eq.~\eqref{eq:misalignment_rotation}. The resulting
    dynamics are propagated with the full rotating-wave Hamiltonian
    Eq.~\eqref{eq:general_rwa_hamiltonian}.
    (a) Full-dynamics infidelity \(1-F\) as a function of the
    inter-manifold basis-mismatch angle.
    (b) Maximum transient mediator excitation
    \(P_{\mathrm{med}}^{\max}\) for the same controls.
    Points denote the 100 three-spin systems, solid lines denote
    medians, and shaded regions show the 10th-90th percentile interval.
    }
    \label{fig:dial_misalignment_audit}
\end{figure*}

Figure~\ref{fig:dial_misalignment_audit} shows that departure from
configuration preservation primarily affects the accuracy of the
engineered interaction. The median fidelity decreases from
\(0.9971\) at \(\eta=0^\circ\) to \(0.9858\) at \(5^\circ\),
\(0.9529\) at \(10^\circ\), \(0.9013\) at \(15^\circ\), and
\(0.8696\) at \(17.5^\circ\). In contrast, the median maximum transient
mediator excitation changes only from approximately \(1.29\%\) at
\(\eta=0^\circ\) to \(1.21\%\) at \(17.5^\circ\).

This benchmark therefore quantifies how the performance of a DIAL control
constructed from the configuration-preserving model changes as the actual
dynamics depart from that model. The fidelity loss occurs while real
mediator excitation remains strongly suppressed, showing that the dominant
effect of the basis mismatch is the appearance of register-changing
transition amplitudes rather than a loss of dispersive excitation
suppression. These amplitudes are present in
Eq.~\eqref{eq:general_rwa_hamiltonian} but absent from the
configuration-preserving dispersive model underlying
Eq.~\eqref{eq:linear_hamiltonian_engineering_system}.

\section{Derivation of Proposition~\ref{prop:general_rwa}}
\label{app:proof_rwa_hamiltonian}

We begin from the projected Hamiltonian given in
Proposition~\ref{prop:general_lab_frame},
\begin{equation}
H_0(\eta)
=
\ket{g}\!\bra{g}\otimes H_g
+
\ket{e}\!\bra{e}\otimes H_e(\eta) .
\label{eq:app_rwa_projected_static}
\end{equation}

To describe the general situation in which the register
quantization axes differ between the two mediator manifolds, we
introduce a one-parameter family of excited-manifold Hamiltonians
$H_e(\eta)$, where the parameter $\eta$ quantifies the degree of
misalignment. The aligned limit corresponds to
\begin{equation}
H_e(0)=\widetilde H_e,
\end{equation}
where $\widetilde H_e$ shares the eigenbasis of $H_g$. No particular
microscopic form of the interpolation is assumed.

The register eigenstates are therefore defined by
\begin{equation}
H_g\ket{\alpha_g}
=
E_g(\alpha)\ket{\alpha_g},
\end{equation}
and
\begin{equation}
H_e(\eta)\ket{\beta_e(\eta)}
=
E_e(\beta,\eta)\ket{\beta_e(\eta)},
\end{equation}

The overlap between the two register eigenbases is
\begin{equation}
M_{\beta\alpha}(\eta)
=
\braket{\beta_e(\eta)}{\alpha_g},
\label{eq:app_overlap_matrix}
\end{equation}
which reduces to
\begin{equation}
M_{\beta\alpha}(0)
=
\delta_{\beta\alpha}
\end{equation}
in the aligned limit.

The corresponding transition frequencies are
\begin{equation}
\Lambda_{\beta\alpha}(\eta)
=
E_e(\beta,\eta)-E_g(\alpha),
\end{equation}
and the detuning of drive tone $k$ is
\begin{equation}
\Delta_{k,\beta\alpha}(\eta)
=
\omega_k-\Lambda_{\beta\alpha}(\eta).
\end{equation}

The transverse drive is
\begin{equation}
H_{\mathrm{drive}}(t)
=
\sum_k
\Omega_k(t)
\cos(\omega_k t+\varphi_k)
(\sigma_++\sigma_-)
\otimes
\mathds{1}_{\mathrm{reg}},
\end{equation}
where
\begin{equation}
\sigma_+
=
\ket{e}\!\bra{g},
\qquad
\sigma_-
=
\ket{g}\!\bra{e}
\end{equation}
where \(\Omega_k(t)\) is the physical Rabi-frequency amplitude of tone
\(k\). With this lab-frame convention, the rotating-wave coupling
matrix element is \(\Omega_k(t)/2\).

Because the register eigenbases generally differ between the two
mediator manifolds, the transition operator expands as
\begin{equation}
\sigma_+\otimes \mathds{1}_{\mathrm{reg}}
=
\sum_{\alpha,\beta}
M_{\beta\alpha}(\eta)
\ket{e,\beta_e(\eta)}
\!\bra{g,\alpha_g}.
\label{eq:app_transition_operator}
\end{equation}

Passing to the interaction picture generated by \(H_0(\eta)\),
\begin{equation}
H_I(t)
=
e^{iH_0(\eta)t}
H_{\mathrm{drive}}(t)
e^{-iH_0(\eta)t},
\end{equation}
each transition operator acquires the phase
\begin{equation}
e^{iH_0t}
\ket{e,\beta_e(\eta)}
\!\bra{g,\alpha_g}
e^{-iH_0t}
=
e^{i\Lambda_{\beta\alpha}(\eta)t}
\ket{e,\beta_e(\eta)}
\!\bra{g,\alpha_g}.
\end{equation}

Expanding the cosine into exponentials produces terms oscillating at
$\Lambda_{\beta\alpha}(\eta)-\omega_k$ and
$\Lambda_{\beta\alpha}(\eta)+\omega_k$. Under the rotating-wave
approximation, the rapidly oscillating counter-rotating terms are
discarded, leaving
\begin{equation}
\begin{aligned}
H_I^{(\mathrm{RWA})}(t;\eta)
&=
\sum_{k,\alpha,\beta}
\Bigg[
\frac{\Omega_k(t)}{2}
M_{\beta\alpha}(\eta)
\\
&\qquad\times
e^{-i\Delta_{k,\beta\alpha}(\eta)t+i\varphi_k}
\\
&\qquad\times
\ket{e,\beta_e(\eta)}
\!\bra{g,\alpha_g}
+
\mathrm{h.c.}
\Bigg],
\end{aligned}
\end{equation}
which is given in Proposition~\ref{prop:general_rwa}.

The parameter $\eta$ therefore generates a continuous family of
rotating-wave Hamiltonians interpolating between the aligned and
misaligned regimes. Throughout the main text, the dispersive theory is
developed in the aligned limit,
\begin{equation}
M_{\beta\alpha}(0)
=
\delta_{\beta\alpha},
\end{equation}
for which
\begin{equation}
\Delta_{k,\beta\alpha}(0)
=
\Delta_{k\alpha},
\end{equation}
and the interaction Hamiltonian reduces to the
configuration-preserving form used in
Proposition~\ref{prop:general_dispersive}. In later numerical studies, specific choices of the family
$H_e(\eta)$ provide a controlled way to assess the performance of DIAL
controls as the full dynamics depart from the configuration-preserving
limit.

\section{Derivation of Proposition~\ref{prop:general_dispersive}}
\label{app:proof_dispersive_effective_hamiltonian}

We derive the effective Hamiltonian generated by off-resonant driving in
the dispersive regime. The derivation follows the adiabatic elimination
formalism of Brion \textit{et al.}~\cite{brion2007adiabatic}, adapted to
the rotating-wave Hamiltonian of
Proposition~\ref{prop:general_rwa}. As shown by
Brion \textit{et al.}, this procedure is equivalent, to second order in
the small parameter $\Omega/\Delta$, to a Schrieffer-Wolff
transformation \cite{schrieffer1966relation, bravyi2011schrieffer}.

Proposition~\ref{prop:general_rwa} defines a one-parameter family of
interaction-picture Hamiltonians
\begin{equation}
H_I^{(\mathrm{RWA})}(t;\eta),
\end{equation}
which continuously interpolates between aligned and misaligned register
frames. In this appendix we specialize to the aligned member of this
family,
\begin{equation}
\eta=0,
\end{equation}
for which the overlap matrix becomes
\begin{equation}
M_{\beta\alpha}(0)
=
\delta_{\beta\alpha}.
\end{equation}
The interaction Hamiltonian therefore becomes configuration preserving,
\begin{equation}
\begin{aligned}
H_I(t)
&=
\sum_{\alpha,k}
\Bigg[
\frac{\Omega_k(t)}{2}
e^{-i\Delta_{k\alpha}t+i\varphi_k}
\sigma_+
\\
&\qquad+
\frac{\Omega_k^*(t)}{2}
e^{+i\Delta_{k\alpha}t-i\varphi_k}
\sigma_-
\Bigg]
\otimes P_\alpha .
\end{aligned}
\label{eq:app_disp_rwa}
\end{equation}
where
\begin{equation}
P_\alpha
=
\ket{\alpha}\!\bra{\alpha},
\qquad
\Delta_{k\alpha}
=
\omega_k-\Lambda_\alpha.
\end{equation}

Since
\begin{equation}
P_\alpha P_{\alpha'}
=
\delta_{\alpha\alpha'}P_\alpha,
\end{equation}
the Hilbert space decomposes into independent subspaces labelled by the
register configuration $\alpha$. Each configuration therefore behaves as
an independent driven two-level system.

For clarity we first consider a single register configuration and a
single drive tone, suppressing the indices \((k,\alpha)\). From
Eq.~\eqref{eq:app_disp_rwa}, the interaction-picture Hamiltonian for
this two-level subspace is
\begin{equation}
H_I(t)
=
\frac{\Omega(t)}{2}
e^{-i\Delta t+i\varphi}
\sigma_+
+
\frac{\Omega^*(t)}{2}
e^{+i\Delta t-i\varphi}
\sigma_- .
\label{eq:single_configuration_rwa}
\end{equation}
The factor \(1/2\) is the rotating-wave matrix element associated with
the lab-frame drive convention of Eq.~\eqref{eq:lab_frame_drive}.

We assume the dispersive regime
\begin{equation}
\left|
\frac{\Omega(t)}{2\Delta}
\right|
\ll
1,
\label{eq:dispersive_condition}
\end{equation}
or, equivalently at the level of scaling, \(|\Omega(t)|\ll|\Delta|\), so
that the mediator remains only virtually excited.

Writing the state as
\begin{equation}
\ket{\psi(t)}
=
c_e(t)\ket{e}
+
c_g(t)\ket{g},
\end{equation}
the Schrödinger equation gives
\begin{align}
i\dot c_e(t)
&=
\frac{\Omega(t)}{2}
e^{-i\Delta t+i\varphi}
c_g(t),
\\
i\dot c_g(t)
&=
\frac{\Omega^*(t)}{2}
e^{+i\Delta t-i\varphi}
c_e(t).
\label{eq:coupled_equations}
\end{align}

The first equation may be formally integrated,
\begin{equation}
c_e(t)
=
-i
\int_0^t
dt'\,
\frac{\Omega(t')}{2}
e^{-i\Delta t'+i\varphi}
c_g(t').
\label{eq:formal_integral}
\end{equation}

Because the detuning is large, the exponential oscillates rapidly on the
timescale \(1/|\Delta|\), whereas both \(\Omega(t)\) and \(c_g(t)\) vary
only slowly. Following the adiabatic elimination procedure of
Ref.~\cite{brion2007adiabatic}, we integrate by parts and neglect terms
of order \(\mathcal{O}(\Omega/\Delta^2)\), obtaining
\begin{equation}
c_e(t)
\simeq
\frac{\Omega(t)}{2\Delta}
e^{-i\Delta t+i\varphi}
c_g(t),
\label{eq:adiabatic_excited}
\end{equation}
up to rapidly oscillating boundary terms.

The excited-state amplitude is therefore
\begin{equation}
c_e
=
\mathcal{O}
\!\left(
\frac{\Omega}{2\Delta}
\right),
\end{equation}
confirming that the mediator is only virtually populated.

Substituting Eq.~\eqref{eq:adiabatic_excited} into the second equation
of Eq.~\eqref{eq:coupled_equations} gives an energy shift of magnitude
\begin{equation}
\left|\delta E\right|
=
\frac{|\Omega(t)|^2}{4\Delta}.
\end{equation}
Repeating the calculation beginning from the excited state gives the
opposite dispersive shift. The mediator transition frequency therefore
acquires the ac Stark correction
\begin{equation}
\delta\Lambda^{(\mathrm{Stark})}
=
\frac{|\Omega(t)|^2}{2\Delta}.
\label{eq:single_stark}
\end{equation}

Restoring the register index gives
\begin{equation}
\delta\Lambda_{\alpha}^{(\mathrm{Stark})}(t)
=
\frac{|\Omega(t)|^2}
{2\Delta_\alpha}.
\end{equation}

Finally we restore the full multi-tone drive. Since the interaction
Hamiltonian is linear in the applied tones,
\begin{equation}
H_I
=
\sum_k
H_I^{(k)},
\end{equation}
second-order perturbation theory produces both diagonal contributions
$H_I^{(k)}H_I^{(k)}$ and mixed terms
$H_I^{(k)}H_I^{(l)}$ for $k\neq l$. The mixed terms oscillate at the beat
frequencies $|\omega_k-\omega_l|$ and average to zero provided these are
large compared with the induced Stark rates and the inverse pulse
timescale. Under this condition the Stark shifts add linearly,
\begin{equation}
\delta\Lambda_{\alpha}^{(\mathrm{Stark})}(t)
=
\sum_k
\frac{|\Omega_k(t)|^2}
{2\Delta_{k\alpha}}.
\label{eq:multitone_stark}
\end{equation}

Collecting the configuration-dependent Stark shifts into a single
operator gives
\begin{equation}
H_{\mathrm{eff}}(t)
=
\frac12
\sigma_z
\otimes
\sum_\alpha
\delta\Lambda_{\alpha}^{(\mathrm{Stark})}(t)
P_\alpha,
\end{equation}
which is Eq.~\eqref{eq:general_dispersive_heff} and therefore proves
Proposition~\ref{prop:general_dispersive}.

The effective Hamiltonian depends only on the configuration-dependent
detunings and is independent of the microscopic origin of the mediator
transition frequencies. In the main text it provides the reference
effective model corresponding to the aligned member $(\eta=0)$ of the
family introduced in Proposition~\ref{prop:general_rwa}. The numerical
study of Appendix~\ref{app:misalignment_audit} then applies controls
constructed from this aligned model to the full dynamics at finite
$\eta$, thereby quantifying their performance under increasing
inter-manifold basis mismatch.

\end{document}